\documentclass{article}

\usepackage{amsthm,amsmath,amsfonts,braket,algorithm,graphicx, braket,balance}
\usepackage{fullpage,authblk}

\usepackage{cite}

\theoremstyle{definition} \newtheorem{define} {Definition} [section]
\theoremstyle{definition} 
\newtheorem {theorem} {Theorem}
\newtheorem {corollary} {Corollary}[section]
\newtheorem {lemma} {Lemma}

\newcommand{\kbb}[1]{{\left[#1\right]}}
\newcommand{\kb}[1]{{\left[#1\right]}}

\newcommand{\al}{\mathcal{A}}

\newcommand{\trd}[1]{\left|\left| #1 \right| \right|}

\newcommand{\repeater}{\mathcal{R}}

\newcommand{\st}{\text{ } : \text{ }}

\newcommand{\dc}{\sim_\delta}

\newcommand{\Hmin}{H_\infty}

\newcommand{\bit}{\texttt{bt}}
\newcommand{\phase}{\texttt{ph}}
\newcommand{\B}{B}
\newcommand{\ideal}{\upsilon}

\newcounter{AssumptionCount}
\newcommand{\assumption}[0]{$ $ \newline\textbf{Assumption \arabic{AssumptionCount}: }\refstepcounter{AssumptionCount}}

\usepackage{xcolor}

\newcommand{\leakEC}{\texttt{leak}_{EC}}

\newcommand{\net}{\mathcal{N}}

\floatname{algorithm}{Protocol}

\begin{document}
\title{Security of Partially Corrupted Repeater Chains}
\author{
Adrian Harkness\footnote{AH: Industrial and Systems Engineering, Lehigh University, Bethlehem PA USA. Email: \texttt{adh323@lehigh.edu}}, 
Walter O. Krawec\footnote{WOK: School of Computing, University of Connecticut, Storrs CT USA. Email: \texttt{walter.krawec@uconn.edu}}, 
Bing Wang\footnote{BW: School of Computing, University of Connecticut, Storrs CT USA. Email: \texttt{bing@uconn.edu}}}

\date{}

\maketitle

\begin{abstract}
Quantum Key Distribution allows two parties to establish a secret key that is secure against computationally unbounded adversaries.  To extend the distance between parties, quantum networks, and in particular repeater chains, are vital.  Typically, security in such scenarios assumes the absolute worst case: namely, an adversary has complete control over all repeaters and fiber links in a network and is able to replace them with perfect devices, thus allowing her to hide her attack within the expected natural noise.  In a large-scale network, however, such a powerful attack may be infeasible.  In this paper, we analyze the case where the adversary can only corrupt a contiguous subset of a repeater chain connecting Alice and Bob, while some portion of the network near Alice and Bob may be considered safe from attack (though still noisy).  We derive a rigorous finite key proof of security assuming this attack model and show that improved performance and noise tolerances are possible.
\end{abstract}

\section{Introduction} \label{sec:intro}

Quantum key distribution (QKD) allows for the establishment of secure secret keys between two parties, Alice and Bob, the security of which is guaranteed by the laws of physics.  This is unlike classical public key cryptography which necessarily requires computational assumptions placed on the adversary in order to achieve security.  See \cite{Scarani09:QKD-survey,Pirandola20:QKD-advances,amer2021introduction} for a general survey on QKD.

One of the main limiting factors of QKD performance is distance---since transmissivity decays  exponentially with distance, achieving efficient QKD between two parties that are far away from each other remains a tremendous challenge. One promising solution for long-distance QKD is through quantum networks, e.g., 
\emph{quantum internet}~\cite{kimble2008quantum,caleffi2018quantum,wehner2018quantum}. Quantum networks consist of quantum repeaters that are capable of creating shared end-to-end entanglement between parties,  even if the
repeaters are controlled by an adversary.
As such, they provide a much stronger security guarantee than the current day \emph{trusted node networks}~\cite{toliver2003experimental,le2007stochastic,wen2009multiple,tanizawa2016routing,yang2017qkd,mehic2019novel,li2020mathematical,tysowski2018engineering}, where the trusted nodes
must be trusted.

In almost all QKD performance analyses, the security of the system assumes the absolute worst case, namely that the adversary controls the entire region outside of Alice and Bob's labs.  This implies that the adversary controls all repeaters and fiber links in the entire network, and can even replace them with ideal, noiseless devices, and thus ``hide'' within the expected natural noise.  Considering that such networks are meant to allow for long-distance QKD operation, this is an unrealistic scenario.  In any realistic operational scenario, it is likely that an adversary can only control a strict subset of repeaters and fiber links.  Furthermore, it is also realistic to assume that an adversary can only control a contiguous section of the network (i.e., not scattered, unconnected, repeaters, but instead a connected region of the network, based on the location of the attacker).  

In this work, we consider the above 
realistic {\em partially corrupted} network scenario. As the first step, we consider the simplest form of quantum repeater networks, a {\em quantum repeater chain},  which consists of two end users (Alice and Bob) and a sequence of quantum repeaters connecting the end users; see Figure~\ref{fig:network}.  These repeaters perform Bell swap operations to create end-to-end entanglement between two distant parties (we discuss the details of their operation later).  Once this end-to-end entanglement is established, parties can run the E91~\cite{Ekert91:E91} QKD protocol (the entanglement based version of BB84 \cite{QKD-BB84}) to establish a shared secret key. As we shall see, even in this simple topology of a repeater chain, analyzing security assuming partial corruption is a tremendous challenge. Our approach and results provide insights that can help to analyze the security of more complex quantum networks. 

Naturally, a partially corrupted network 
is an assumption; however, unlike classical key distribution which requires \emph{computational} assumptions, this assumption on the attack model is grounded in physical limitations of the adversary.  See Section \ref{sec:network-assumptions} for a more formal description of our assumptions and security model.  So far, security analyses of QKD networks assuming alternative attack models, such as this, have received very little attention (see Section \ref{sec:prior-work} for a discussion on prior work). Yet analyzing such scenarios is important for a variety of reasons.  For instance, we show that drastically improved key-rates and noise tolerances are possible (as we show in Section \ref{sec:eval}), potentially allowing for early QKD networks to perform at more optimistic levels.  Also, by having rigorous proofs of security that can handle alternative, perhaps more realistic, attack models, users of early QKD network systems can have guidelines on how much of the network, and exactly which locations, need to be protected physically, in order to achieve a certain desired key-rate.

In this work,
we analyze the performance and security of the E91 protocol operating on a repeater chain consisting of $c$ repeaters, where the adversary is allowed to control a subset of contiguous repeaters and links.  Alternatively, one may consider a repeater chain where some of the repeaters and links near Alice and Bob are trusted and better secured.  We assume that users can upper-bound the number of adversarial repeaters or, alternatively, can lower-bound the number of honest repeaters in the network.  There are two motivating examples as to why this is a reasonable assumption.  First, as discussed above, it is unrealistic that an adversary can control the entire, lengthy, repeater chain and replace the entire network with adversarial devices, performing coherent attacks across a large distance.  Second, it is likely that at least some repeaters near Alice and Bob will be placed in a secure and trusted location, which an adversary cannot easily gain access to (similar to how trusted nodes are considered physically secure - however securing a repeater is even easier as it never stores the secret key).  Thus, one can assume that the first $k_A$ repeaters connected to Alice are secure while the last $k_B$ repeaters connected to Bob are secure, leaving the middle $c - k_A-k_B$ repeaters as potentially under the control of the adversary; see Figure \ref{fig:network}. Note that it may be $k_A = 0$ and/or $k_B = 0$.  If both are zero (i.e., if there are no trusted repeaters), then our key-rate result converges asymptotically to the standard BB84 expression.

\begin{figure}
   \centerline{\includegraphics[width=.75\linewidth, trim = 1.0cm 17.5cm 1.cm 1.1cm, clip]{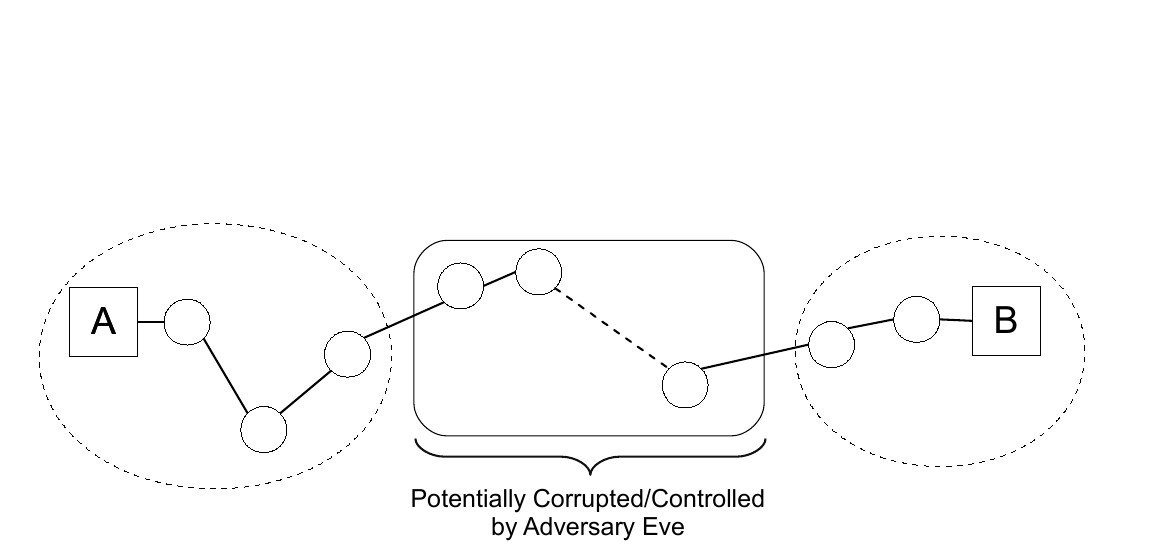}}
\caption{An example repeater chain where the repeaters (solid circles) and links near parties Alice and Bob, (those within the dashed circles), may be considered ``safe'' or trusted, while all other repeaters and links may, or may not, be adversarial.  Note that even though the repeaters near Alice and Bob may be trustworthy, they are still noisy. }
\label{fig:network}
\end{figure}

Of course, even though there may be multiple honest repeaters and fiber links connecting them, these honest sub-networks are still noisy, and will still introduce detectable noise into the final shared entangled pairs.  Thus, when Alice and Bob run a QKD protocol, the observed noise is a function both of Eve's attack and the natural noise in the trusted sub-network.  Therefore, for a given observed noise level $Q$, one would expect that Eve's information is not nearly as high as it would be in the standard assumption case, where all noise is the result of an attack.  However, formalizing this in the finite key setting, where a bound on the quantum min entropy \cite{renner2008security} is required, is non-trivial.  With min entropy, one must take into account that it is in some ways a ``worst-case'' entropy and, so, we must be careful when analyzing the system that Eve is not always able to ``hide'' in the natural noise.  We must also deal with finite sampling imprecisions, and also the fact that Eve can interact non-trivially with the honest repeater network.  Finally, we must also take into account that the repeater network, including the adversary's portion of it, must send classical messages to users of the protocol, in order for them to apply a correcting Pauli gate.  Taken together, these issues make finite key analyses a challenge.

Finite key analyses, however, are vital to understanding the potential performance of a quantum system.  Asymptotic analyses are highly interesting, and useful, as theoretical upper-bounds; however finite key scenarios, where Alice and Bob only utilize the network for a finite amount of rounds, are important for understanding the potential performance in more realistic scenarios.  Our work shows how to bound the quantum min entropy between Alice and the adversary Eve in this network scenario, thus providing us with a bound on the finite key-rate of the system under this attack model.  We develop a novel proof technique for this scenario, taking advantage of a sampling-based framework introduced by Bouman and Fehr in \cite{bouman2010sampling}, along with proof techniques used for sampling based entropic uncertainty relations \cite{yao2022quantum,krawec2019quantum}. We restrict our attention to noisy but lossless channels as this already presents a large challenge; despite this, we suspect our proof techniques may be extended to deal with lossy channels also, potentially using decoy state methods \cite{QKD-first-decoy-1,QKD-first-decoy-2,QKD-first-decoy-3}.  However, we do not assume uniform noise in the network; some honest links may be noisier than others.  We also do not make any assumptions on the adversary's attack within the corrupted sub-network; i.e., it can be any arbitrary general/coherent attack.

We make several contributions in this work.  First, we show a quantum min entropy bound for the setting where natural and adversarial noise are mixed in a quantum repeater chain.  Such a bound allows us to derive key-rate expressions in the practical finite key setting.  To prove our new bound, we develop several new techniques which may be broadly applicable outside of this application domain.  At a high level, our new results claims that the quantum min entropy, denoted $\Hmin^\epsilon(A|E)$ is bounded by:
\begin{equation}
  \Hmin^\epsilon(A|E) \ge n(1 - h(Q-Q_{\net}+\delta)),
\end{equation}
where $h(x)$ is the binary Shannon entropy, $n$ is the number of network rounds used (after sampling), $Q$ is the observed $X$ basis noise, $Q_{\net}$ is a function of the honest network noise (which we assume in this paper that Alice and Bob may at least lower-bound), and $\delta$ results from finite sampling imperfections.  Our proof takes into account all finite sampling artifacts and imprecisions, allowing users to immediately evaluate key-rates and optimize over user parameters.  Our full result is stated in Theorem \ref{thm:main}.

We use our expression to derive finite key-rate expressions, and also asymptotic key-rates.  We evaluate our results in a variety of settings showing that significantly improved key-rates are possible, compared to standard security models consisting of entirely adversarial networks.  While that result is not surprising, showing it rigorously is a challenging, but important, problem which we solve in this work.

Finally, our proof method may be broadly applicable to other application domains within quantum cryptography.  We build on the quantum sampling framework of Bouman and Fehr and introduce new methods to derive min entropy expressions for systems that are only partially under the control of the adversary.

\subsection{Prior Work}\label{sec:prior-work}

We are not the first to consider physical assumptions in the security model of QKD.  Much work has been done, for instance, in assuming Eve is bounded in her storage abilities, either in quantity of storage bits \cite{damgaard2008cryptography}, or quality of storage memory \cite{wehner2008cryptography}.

We are also not the first to consider a security model based on the communication setup.  For instance, several recent papers have considered alternative security models for satellite communication, placing reasonable assumptions on the adversary's capabilities given the channel conditions.  In particular, these references take advantage of the fact that satellite communication requires line of sight and it is infeasible for an adversary to completely control the freespace channel between the satellite transmitter and the ground station.  In \cite{ghalaii2023satellite}, a new ``bypass channel'' model is introduced which models the practical assumption that an adversary can only capture a portion of the transmitted photons while others will bypass the adversary and arrive at the receiver un-attacked.  Ref \cite{vergoossen2019satellite} took this further and argued that attacks against satellite QKD can be detected through classical means, and defined ``photon key distribution'' protocols to improve performance.  Other references \cite{pan2020secret,vazquez2021quantum} have considered security of QKD protocols, particularly freespace ones such as satellite communication, operating over wiretap channels \cite{wyner1975wire}.  These are all assumptions placed on the adversary, based on reasonable practical constraints on any attack against QKD.  In our work we place what we consider reasonable assumptions on an adversary, based on the impracticality of attacking an entire large-scale network simultaneously.  Similar to the work cited above, these assumptions allow for improved performance of the underlying system - though it is up to users of the system to decide if they are comfortable with the assumption.  Indeed, users can always revert back to the standard security model (though, in that case, they can no longer take advantage of the improved performance).

Perhaps the closest work to ours is found in \cite{mertz2013quantum,graifer2023quantum}.  Both of these sources investigated the performance of QKD where some of the observed channel noise is assumed to be honest or natural noise (also called trusted noise in some references), while some is adversarial.  Both references, however, only considered point-to-point BB84, not a quantum repeater chain.  Furthermore, \cite{mertz2013quantum} only considered collective attack scenarios and, thus, computed a bound on the von Neumann entropy (note that such an analysis could be promoted to general attacks, though the result is usually not as tight, \emph{in the finite key setting}, as deriving a bound directly on the min entropy as we do in this paper).  We do not assume the adversary is restricted to collective attacks, thus requiring us to derive a bound on the quantum min-entropy, a more challenging prospect.  The second, ref. \cite{graifer2023quantum}, considered a particular ``state replacement'' noise model, where the natural noise in the channel consisted of a state being replaced with a truly mixed state.  This replacement is done for every single state sent, in an i.i.d. manner and the probability of state replacement is known and characterized.  Our work considers repeater chains where some of the repeater chain is considered ``honest'' or safe, but suffers from characterizable noise, while the remainder of the network is considered adversarial.   Our proof must take into account the action of the honest repeater network and the classical messages being passed, which was not a requirement in these previous works.

Other references have considered various ``trusted noise'' scenarios in the discrete variable case.  In \cite{shadman2009optimal}, the six-state BB84 protocol was analyzed where the signal received by Bob is mixed with white noise (which can be added deliberately by the source, Alice, or naturally, such as by natural light interfering with a free-space satellite QKD link).  However, only individual attacks were considered in the asymptotic setting; note that individual attacks are weaker than collective attacks and security against individual attacks does not necessarily imply security against arbitrary, general attacks.  In \cite{mafu2022security}, the authors investigated the performance of BB84 with a particular form of added natural noise, namely collective-rotation noise.  However, the security analysis was only against a particular intercept/resend attack strategy, where the adversary measures incoming signals in either the $Z$ or $X$ basis.  This was followed up recently in \cite{garapo2016intercept} for the six-state BB84 protocol, but again, only for intercept resend attacks.  In \cite{woodhead2014tight}, the benefits of adding noise to an already faulty source were considered and shown to improve BB84.  Finally, in \cite{jung2009attack}, the effects of multiple but independent, adversaries on a single point-to-point BB84 link were considered.  Only the von Neumann entropy was investigated there.

BB84 style protocols were not the only ones to be considered in the trusted-noise scenario.  In \cite{sharma2018decoherence,utagi2020ping}, the so-called Ping-Pong protocol (introduced in \cite{QKD-TwoWay-PingPong-SDC}) was analyzed assuming there was either trusted noise in the channel \cite{sharma2018decoherence} or there was noise added by the source \cite{utagi2020ping}. The Ping-Pong protocol relies on a two-way quantum communication channel, with qubits traveling from Alice, to Bob, then back to Alice.  In both these works, only asymptotic analyses were considered and, thus, bounds on von Neumann entropy.  Furthermore, no quantum repeaters were considered.  Finally, larger scale networks were also considered in \cite{le2007stochastic}, though, there, the network consisted only of trusted nodes (not repeaters) and the security model assumed that trusted nodes were corrupted randomly; the goal of that reference was to route QKD paths randomly so that at least one path went through all honest trusted nodes.  This is different from our work where we are forced to use a single path, thus passing through both the honest, and the dishonest, nodes.

Moving beyond these discrete-variable protocols, several sources have investigated natural and trusted noise in the continuous variable QKD scenario \cite{usenko2010feasibility,usenko2010feasibility,pirandola2021composable,liu2022composable,garcia2009continuous}; see also \cite{laudenbach2018continuous} for more of a survey in practical continuous variable QKD.  However, none of these considered repeater chains and, instead, assumed natural noise in the channel between source and receiver, trusted noise in the devices, or the intentional addition of noise at the source or receiver.

\subsection{Preliminaries}\label{section:notation}

We now introduce some notation and basic definitions we use throughout this work.  We will then discuss some more important properties of quantum min entropy and some basic lemmas which will be used later.

Let $\al_d$ be a $d$-character alphabet which, without loss of generality, we simply assume to be $\al_d = \{0, 1, \cdots, d-1\}$.  Given a word $q \in \al_d^N$ and a subset $t \subset \{1, \cdots, N\}$, we write $q_t$ to be the substring of $q$ indexed by subset $t$ (i.e., $q_t = q_{t_1}\cdots q_{t_{|t|}}$) and we write $q_{-t}$ to mean the substring of $q$ indexed by the complement of $t$.  When $t$ is a singleton $t = \{i\}$ we usually just write $q_i$ to mean the $i$-th character of $q$.  We use $w(q)$ to be the relative Hamming weight of $q$, defined by:
\[
  w(q) = \frac{|\{i \st q_i \ne 0\}|}{|q|}.
\]
Finally, given two real values $x, y \in \mathbb{R}$ and $\delta > 0$, then we write:
\begin{equation*}%
x \dc y
\end{equation*}
if and only if $|x-y| \le \delta$.

Let $P$ be some probability distribution over $\al_d$, with $P(x)$ being the probability of some outcome $x\in\al_d$.  Then, given a word $q \in \al_d^N$, for some $N > 1$, we often write $P(q)$ to mean $P(q) = P(q_1)P(q_2)\cdots P(q_N)$.

If a quantum state (density operator) $\rho$ acts on some Hilbert space $\mathcal{H}_A\otimes\mathcal{H}_B$, we usually write $\rho_{AB}$; we then write $\rho_A$ to mean the state resulting from the partial trace over $B$, namely $\rho_A = tr_B\rho_{AB}$.  This can be extended to multiple subspaces.  Given a pure state $\ket{\psi}$, we write $\kb{\psi}$ to denote $\kb{\psi} = \ket{\psi}\bra{\psi}$.  Given an orthonormal basis $\mathcal{B} = \{\ket{x_0}, \cdots, \ket{x_{d-1}}\}$ and a word $i \in \al_d^N$, we write $\ket{i}^\mathcal{B}$ to mean the word $i$ in the $\mathcal{B}$ basis, namely $\ket{i}^\mathcal{B} = \ket{x_{i_1}}\otimes\cdots \otimes\ket{x_{i_N}}$.  If no basis is specified, we assume the standard computational basis, namely $\ket{i} = \ket{i_1}\otimes\cdots\otimes\ket{i_N}$.  Finally, given $\rho$ and $\sigma$, acting on the same Hilbert space, we write $\trd{\rho - \sigma}$ to be the trace distance of $\rho$ and $\sigma$ defined as: $\trd{\rho - \sigma} = tr\sqrt{(\rho-\sigma)^*(\rho-\sigma)}$, where $A^*$ is the Hermitian adjoint of operator $A$.

$ $\newline
\textbf{Bell Basis Notation: }
We use $\ket{\phi_x^y}$, for $x,y\in\{0,1\}$, to denote the Bell basis states:
\begin{equation}
  \ket{\phi_x^y} = \frac{1}{\sqrt{2}}(\ket{0,x} + (-1)^y\ket{1,\bar{x}}),
\end{equation}
where $\bar{x} = 1-x$.  Later, we will work with multiple Bell states tensored together.  For this, we define the \emph{Bell alphabet set of size $N$} to be:
\begin{equation}
  \B^N = \{(x,y) \in \{0,1\}^N\times\{0,1\}^N\}
\end{equation}
Given an element $i = (x,y) \in \B^N$, we write $i^\bit$ to mean the ``$x$'' portion of the string $i$ while we write $i^\phase$ to be the $y$ portion.  That is, the superscript ``$\phase$'' will denote the ``phase'' element of a Bell state, while ``$\bit$'' will represent the ``bit'' portion.  All subset indexing rules discussed earlier apply to each individual portion of $i$ (e.g., $i^\phase_t$ is the $y$ portion of $i$, but only those indices indexed by $t$).  We then write $i_t$ to mean both $x$ and $y$ portions indexed by $t$, namely $i_t = (x_t, y_t) \in \B^{|t|}$. We write $\ket{\phi_i}$ to mean $\ket{\phi_{i^\bit}^{i^\phase}}$ with:
\begin{equation}
  \ket{\phi_i} = \ket{\phi_{i^\bit}^{i^\phase}} = \ket{\phi_{x_1}^{y_1}}\otimes\ket{\phi_{x_2}^{y_2}}\otimes\cdots\otimes\ket{\phi_{x_N}^{y_N}}.
\end{equation}
(Recall $x_j$ is the $j$'th bit of $x$ and similarly for $y$.)

Finally, we also can add two Bell alphabet elements: given $i,j\in\B^N$ with $i = (x,y)$ and $j = (z,u)$, then we write $i+j$ to mean the addition, coordinate-wise, modulo two, namely: $i+j = (x\oplus z, y\oplus u)$, where the strings $x\oplus z$ and $y\oplus u$ are added bit-wise modulo two.

\subsection{Quantum Min Entropy}
Let $X$ be a random variable taking value $i$ with probability $p_i$.  Then we write $H(X)$ to mean the Shannon entropy of $X$ defined to be $H(X) = -\sum_ip_i\log p_i$ where all logarithms in this paper are base two unless otherwise specified.  If $X$ has only two outcomes, then $H(X) = h(p) = -p\log p - (1-p)\log(1-p)$ where $h(p)$ is the binary entropy function.  For technical reasons later, we define a function $\bar{h}(p)$ by $\bar{h}(p) = h(p)$ if $p < 1/2$ and $\bar{h}(p) = 1$ otherwise.  Thus $\bar{h}(p) \le \bar{h}(p')$ for every $0\le p \le p' \le 1$.

Given a quantum state $\rho_{AE}$, the \emph{conditional quantum min entropy} is defined as \cite{renner2008security}:
\begin{equation}
  \Hmin(A|E)_\rho = \sup_{\sigma_E}\max\left\{\lambda\in\mathbb{R} \st 2^{-\lambda}I_A \otimes \sigma_E - \rho_{AE} \ge 0\right\},
\end{equation}
where the supremum is over all density operators $\sigma_E$ acting on $\mathcal{H}_E$ and where $A \ge 0$ means operator $A$ is positive semi-definite.  Let $\Gamma_\epsilon(\rho) = \{\tau_{AE} \st \trd{\rho_{AE} - \tau_{AE}} \le \epsilon\}$, i.e., the set of all density operators $\epsilon$-close to $\rho_{AE}$ in trace distance.  Then, the \emph{smooth min entropy} is defined \cite{renner2008security} to be:
\begin{equation}
  \Hmin^\epsilon(A|E)_\rho = \sup_{\tau \in \Gamma_\epsilon(\rho)}\Hmin(A|E)_\tau.
\end{equation}

Quantum min entropy is a vital resource in quantum cryptography as it directly relates to how many uniform random bits may be extracted from a quantum state.  Formally, let $\rho_{AE}$ be a \emph{classical quantum} state (cq-state).  That is, it may be written in the form $\rho_{AE} = \sum_{a\in\{0,1\}^N}P_A(a)\kb{a}_A \otimes \rho_E^{(a)}$.  Then, if one chooses a two-universal hash function at random, $f:\{0,1\}^N \rightarrow \{0,1\}^\ell$, disclosing the choice of function to Eve, and hashing  the $A$ register to $f(A)$, then it holds \cite{renner2008security}:
\begin{equation}\label{eq:PA}
  \trd{\rho_{f(A),EF} - I/2^\ell\otimes\rho_{EF}} \le 2^{-\frac{1}{2}(\Hmin^\epsilon(A|E)_\rho - \ell)}+2\epsilon.
\end{equation}
Above, $F$ is the system representing Alice's random choice of hash function $f$, while $f(A)$ is the $\ell$-bit register resulting from hashing $N$-bit register $A$.  Essentially, the above states that, so long as the min entropy in the state $\rho_{AE}$ \emph{before} privacy amplification is high enough, one can extract a random string, of size $\ell$-bits, that is uniform random and also independent of Eve.

There are several important properties of min entropy that will be useful later.  Given a state $\rho_{ABC}$ that is classical in $C$, namely it can be written in the form $\rho_{ABC} = \sum_cp_c\kb{c}\otimes\rho_{AB}^{(c)}$, then it holds that:
\begin{equation}\label{eq:min-entropy-mixed}
  \Hmin(A|B)_\rho \ge \Hmin(A|BC)_\rho \ge \min_c \Hmin(A|B)_{\rho^{(c)}}.
\end{equation}
In particular, the above says that for some mixed state $\rho_{AB} = \sum_cp_c\rho_{AB}^{(c)}$, or a state $\rho_{ABC}$ that is classical in $C$, the min entropy of the entire state can be lower-bounded by the worst-case entropy over the possible sub-events $c$.

We conclude this section with two lemmas that will be useful later.  The first one, below, allows us to bound the smooth min entropy in a particular state after a measurement operation is performed on part of it, if we know the min entropy of a state that is ``close'' to it in trace distance:
\begin{lemma} \label{lemma:cptp-entropy}
  (From \cite{krawec2022security}): Let $\epsilon > 0$, and let $\rho$ and $\sigma$ be quantum states acting on the same Hilbert space such that $\frac{1}{2}\trd{\rho-\sigma}\le\epsilon$.  Let $\mathcal{F}$ be some completely positive trace preserving (CPTP) map (i.e., some quantum operation, or operations) which, on input a quantum state $\tau$, acts as follows:
  \begin{equation}
    \mathcal{F}(\tau) = \sum_xp(x|\tau)\kb{x}_X\otimes\tau_{AE}^{(x)}.
  \end{equation}
  Then, it holds that:
  \begin{equation}
    Pr\left(\Hmin^{4\epsilon + 2\epsilon^{1/3}}(A|E)_{\rho^{(x)}} \ge \Hmin(A|E)_{\sigma^{(x)}}\right) \ge 1 - 2\epsilon^{1/3},
  \end{equation}
  where the probability is over the random outcome $X$ in the states $\mathcal{F}(\rho)$ and $\mathcal{F}(\sigma)$.
\end{lemma}

The next lemma that we will need in our proof, allows us to bound the min entropy of a superposition state in the Bell basis, if a measurement is made on the first qubit of every Bell pair:
\begin{lemma}\label{lemma:bell-entropy}
  (From \cite{krawec2023entropic}, rewritten using our notation): Let $\ket{\psi}_{XE} = \sum_{i\in J}\alpha_i\ket{\phi_{i^\bit}^{i^\phase}}_X\otimes\ket{E_i}_E$ where $J = \{i\in \B^n \st w(i^\phase) \le Q\} \subset \B^n$ (for $Q \in [0,1]$).  Let $\rho_{AE}$ be the state resulting from taking $\ket{\psi}$ and measuring the first particle of every Bell pair in register $X$ in the computational basis (which results in register $A$) while the second particle of every Bell pair is traced out.  This measurement results in post-measured state $\rho_{AE}$.  Then, it holds that:
  \begin{equation}
    \Hmin(A|E)_\rho \ge n\left(1 - \bar{h}(Q)\right)
  \end{equation}
\end{lemma}

\subsection{Quantum Sampling}

To derive a lower-bound on the quantum min entropy, we will take advantage of a quantum sampling framework introduced by Bouman and Fehr in \cite{bouman2010sampling}.  In this section, we review some of the main points of this framework, referring the reader to \cite{bouman2010sampling} for additional details.  The main point of Bouman and Fehr's sampling framework is the ability to promote a classical sampling strategy to a quantum one in such a way that one can argue about the state of the post measured system after sampling a quantum state.

A \emph{classical sampling strategy} for words of length $N$ over some alphabet $\al_d$ is a triple $(P_T, g, r)$ where $P_T$ is a probability distribution over all subsets of $\{1, \cdots, N\}$; $g$ is a \emph{guess function}; and $r$ is a \emph{target function}.  Here, $g, r:\al_d^* \rightarrow \mathbb{R}$.  Typically $g \equiv r$ (which will certainly be the case in this work, where we set $g \equiv r \equiv w$, the Hamming weight function), though this is not required in general.  Given a word $q \in \al_d^N$, the sampling strategy will: (1) choose a subset $t$ according to distribution $P_T$; (2) observe $q_t$ and evaluate $g(q_t)$ (or, will simply observe $g(q_t)$); finally, (3) output this value as a ``guess'' as to the value of $r(q_{-t})$.  That is, the strategy uses $g(q_t)$ (the guess function evaluated on the observed portion of $q$) to guess at the value of some target function evaluated on the unobserved portion of $q$.

Let $\delta >0$, then given a fixed subset $t$, we define the set of \emph{ideal words} to be those words $q$ where the guess, given $q_t$, is always $\delta$-close to the target $r(q_{-t})$.  Formally:
\begin{equation}
  \mathcal{G}_t = \left\{ q \in \al_d^N \st g(q_t) \dc r(q_{-t})\right\}.
\end{equation}
(Recall $x\dc y$ only if $|x-y| \le \delta$.)  From this, the \emph{error probability} of the given sampling strategy is defined to be:
\begin{equation}
  \epsilon^{cl} = \max_{q\in\al_d^N}Pr\left(q\not\in \mathcal{G}_t\right),
\end{equation}
where the probability, above, is over the choice of subset $t$ chosen according to $P_T$.  Thus, given any string $q\in\al_d^N$, the probability that the given classical sampling strategy fails to produce a $\delta$-close guess of the target value on observing $q_t$ is no higher than $\epsilon^{cl}$.

While the above described sampling strategy applies to a classical word, it may be promoted, in a natural way, to a quantum sampling strategy.  Given a quantum state $\ket{\psi}_{AE}$, where the $A$ register lives in some $d^N$ dimensional Hilbert space, the sampling strategy will choose a subset $t$ according to $P_T$ and then \emph{measure} those qudits in $A$ indexed by $t$ in some fixed $d$-dimensional basis $\mathcal{B}$.  This measurement produces a classical output $q_t\in\al_d^{|t|}$ and a quantum post-measured state $\ket{\psi^{t,q_t}}_{A'E}$, which depends on both the subset choice $t$ and the actual observed value $q_t$.  Note that the $A'$ register of the post-measured state, once removing the measured qudits, lives in a $d^{N-|t|}$-dimensional Hilbert space.  The question becomes: what can be said of $\ket{\psi^{t,q_t}}_{A'E}$?

Define the space of \emph{ideal states} for subset $t$ with respect to the fixed (but arbitrary) $d$-dimensional basis $\mathcal{B}$ as follows:
\begin{equation}
  \text{span}\left(\mathcal{G}_t\right)\otimes\mathcal{H}_E = \text{span}\left\{\ket{q}^{\mathcal{B}} \st q \in \mathcal{G}_t\right\}\otimes\mathcal{H}_E.
\end{equation}
An \emph{ideal state}, $\ket{\ideal^t}$, is defined to be one which lives in this space.  Note that, if a basis measurement in the $\mathcal{B}$ basis is made of $\ket{\ideal^t}$ in subset $t$, producing outcome $x\in\al_d^{|t|}$, then it is guaranteed that the post-measured state can be written in the form:
\begin{equation}
  \ket{\ideal^t_x} = \sum_{i \in J_x}\alpha_i\ket{i}^{\mathcal{B}}\ket{E_i},
\end{equation}
where:
\begin{equation}
  J_x = \left\{i \in \al_d^{N-|t|} \st g(x) \dc r(i)\right\}
\end{equation}
Notice that, if the state given is an ideal state with respect to subset $t$ and if the sampling strategy actually chooses $t$ to sample, then the post-measured state is well behaved.  Of course, given an arbitrary state $\ket{\psi}_{AE}$, this is not guaranteed.  However, Bouman and Fehr's main result, stated in Theorem \ref{thm:sample} below, says that, roughly, $\ket{\psi}_{AE}$ should behave like an ideal state, on average over the subset choice.

\begin{theorem}\label{thm:sample}
  (From \cite{bouman2010sampling}, though reworded here for our application): Let $\delta > 0$ and $\ket{\psi}_{AE}$ be an arbitrary quantum state where the $A$ register consists of $N$ qudits each of dimension $d$.  Let $\mathcal{B}$ be an arbitrary $d$-dimensional orthonormal basis.  Then, given a classical sampling strategy $(P_T, g, r)$ with failure probability $\epsilon^{cl}$, there exists a collection of ideal states $\{\ket{\ideal^t}\}_t$, indexed over every subset $t$, such that $\ket{\ideal^t} \in \text{span}\left(\mathcal{G}_t\right)\otimes\mathcal{H}_E$ (where $\text{span}\left(\mathcal{G}_t\right)$ is defined with respect to basis $\mathcal{B}$) and:
  \begin{equation}
    \frac{1}{2}\trd{\sum_tP_T(t)\kb{t}\otimes\kb{\psi} - \sum_tP_T(t)\kb{t}\otimes\kb{\ideal^t}} \le \sqrt{\epsilon^{cl}}.
  \end{equation}
\end{theorem}
\begin{proof}
  For a proof, see \cite{bouman2010sampling}; to see that our rewording of their main result follows from Bouman and Fehr's work, the reader is also referred to \cite{yao2022quantum}.
\end{proof}

To conclude this section, we will introduce the classical sampling strategy we will use later which we denote here by $\Psi_4$.  It operates on the four-dimensional alphabet $\B^N$ and is defined as follows: $P_T$ will choose a random subset $t$ of size $m\le N/2$, uniformly at random from all subsets of $\{1, \cdots, N\}$ of size $m$.  Then, the guess function and target functions are the Hamming weight of the phase component of the word $q \in \B^N$.  Namely $g(q_t) = w(q^\phase_t)$ and $r(q_{-t}) = w(q^\phase_{-t})$.

To bound the error probability of this strategy, we will actually need to introduce an alternative strategy defined and analyzed in \cite{bouman2010sampling}, for two character alphabets which we denote $\Psi_{hw}$.  Namely, given a word $q \in \{0,1\}^N$, $P_T$ will choose a uniform random subset of size $m\le N/2$, observe the relative Hamming weight of $q_t$, namely $g(q_t) = w(q_t)$ and use this as a guess for the target value $r(q_{-t}) = w(q_{-t})$ (i.e., the target value is the Hamming weight of the unobserved portion).  It was shown in \cite{bouman2010sampling} that the error probability for this strategy is upper-bounded by $\epsilon_{hw}^{cl}$ defined to be:
\begin{equation}\label{eq:hw-sample}
  \epsilon_{hw}^{cl} \le 2\exp\left(-\delta^2\frac{mN}{N+2}\right).
\end{equation}

Using this, we can prove the following lemma, bounding the error probability of the sampling strategy $\Psi_4$:
\begin{lemma}\label{lemma:sample-used}
  Given the sampling strategy $\Psi_4$ described above, for $m < N/2$, it holds that:
  \begin{equation}
    \epsilon^{cl} \le 2\exp\left(-\delta^2\frac{mN}{N+2}\right).
  \end{equation}
\end{lemma}
\begin{proof}
Let $\mathcal{G}_t^{(4)}$ be the set of good words induced by sampling strategy $\Psi_4$ and let $\mathcal{G}_t^{hw}$ be the set of good words for sampling strategy $\Psi_{hw}$.  Pick $q \in \B^N$ and let $\widetilde{q} = q^\phase$.  Then it is obvious that $q \not\in \mathcal{G}_t^{(4)} \iff \widetilde{q}\not\in \mathcal{G}_t^{hw}$.  Thus, $Pr(q \not\in \mathcal{G}_t^{(4)}) = Pr(\widetilde{q}\not\in\mathcal{G}_t^{hw})$.  Since $q$ was arbitrary, and using Equation \ref{eq:hw-sample}, the result follows.
\end{proof}


\section{Network and Security Model}\label{sec:network-model}

We consider a repeater chain topology in this work consisting of $c$ repeaters, denoted $\repeater_1, \cdots, \repeater_c$, chained in sequence connecting two users Alice and Bob as shown in Fig.~\ref{fig:chain}(a).  A repeater in our network is a basic device with two quantum storage ports, one connected to each neighbor.  These devices are capable of creating Bell pairs and sending one particle to a neighbor while storing the other in 
quantum memory; receiving quantum states and storing them in the corresponding storage port; performing Bell measurements on the two qubits in storage; and finally, sending and receiving classical messages.  We assume a noisy but lossless quantum communication model in this work, leaving the lossy case to future work.

\begin{figure}[t]
   \centerline{\includegraphics[width=3.5in, trim = 0.cm 0.0cm 0.cm 0.1cm, clip]{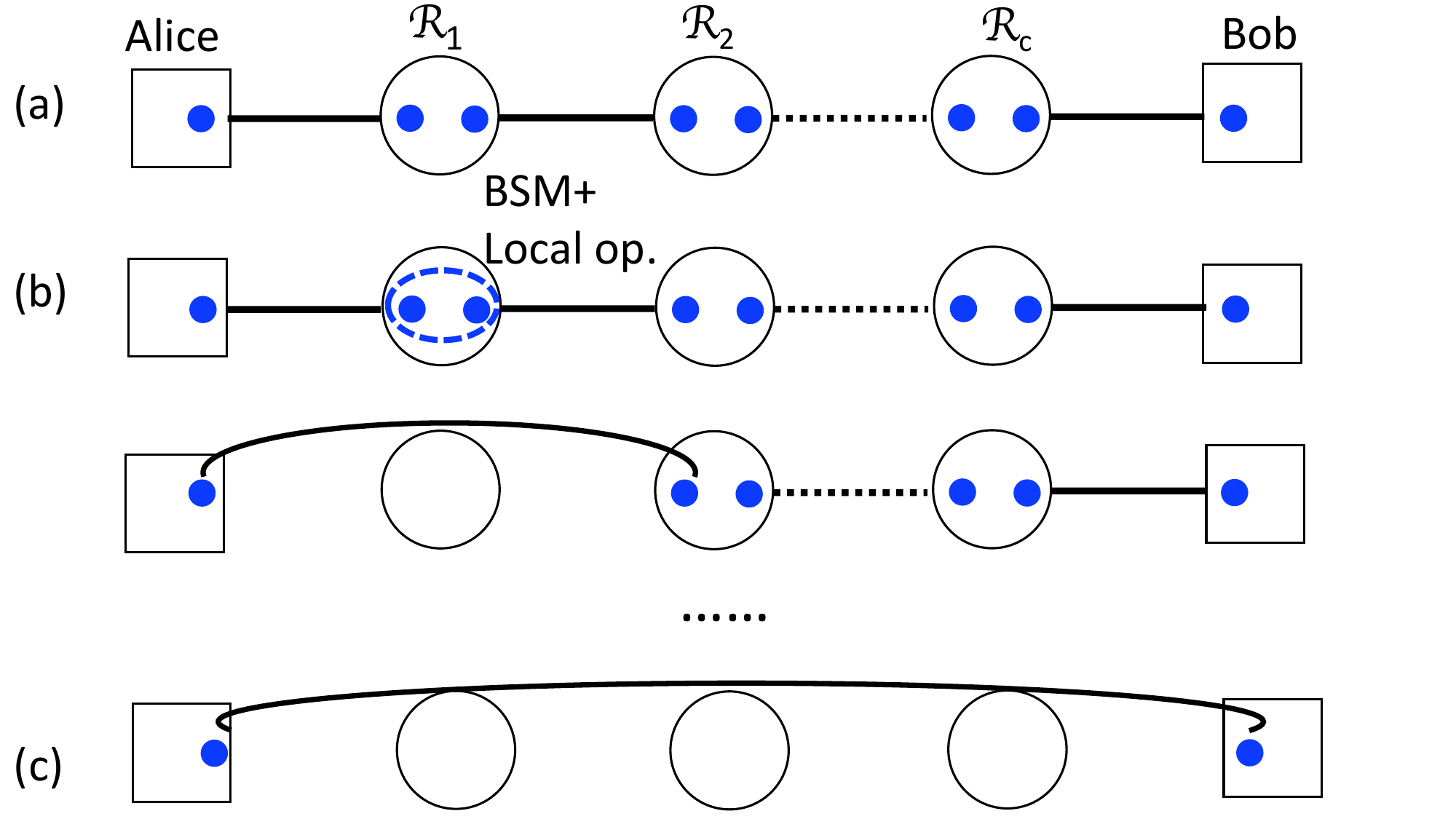}}
\caption{{\small Illustration of the operation of the repeater chain. (a) Link-level distribution of Bell pairs. (b) Bell state measurement (BSM) and local Pauli gate operation at repeater $\repeater_1$ to create an entanglement between Alice and repeater $\repeater_2$. (c) End-to-end entanglement between Alice and Bob. }}
\label{fig:chain}
\end{figure}

If the entire repeater chain is honest, the network will perform the following operations on each round (refer also to Fig.~\ref{fig:chain}):
\begin{enumerate}
\item First, repeater $\repeater_1$ will create two Bell pairs and send one particle to Alice and one particle to $\repeater_2$.  The other two particles (one from each pair) are stored in the corresponding storage ports of $\repeater_1$.
\item Repeater $\repeater_2$ will then store the received particle from $\repeater_1$ while creating a new Bell pair and sending one particle to $\repeater_3$.
\item Repeaters continue to distribute link-level Bell pairs until all repeaters have two particles each while Alice and Bob have one particle each.
\item While the above is happening, repeater $\repeater_1$ will, as soon as possible, perform a Bell measurement on both particles in its memory, thus creating, ideally, an entangled pair between Alice and repeater $\repeater_2$; see Fig. \ref{fig:chain}(b).  Furthermore, the outcome of this measurement (which we denote as simply 
$``x,y"$ if Bell state $\ket{\phi_x^y}$ is observed)
is sent to Alice.
\item Alice will, on receipt of the classical message from $\repeater_1$,  apply an appropriate Pauli gate to her particle in the right storage port (received from $\repeater_1$ and which, should ideally now, be entangled with $\repeater_2$).  This should, in the absence of noise, ensure that Alice and $\repeater_2$ have the Bell pair $\ket{\phi_0^0}$.

  \item As soon as $\repeater_2$ has two particles in its storage port (namely, as soon as it sends a particle to $\repeater_3$), it will perform a Bell measurement itself, reporting the outcome to Alice, who applies the correct Pauli gate as before.
  \item The above continues until, finally, the last repeater 
  $\repeater_c$ performs a Bell measurement, reporting the outcome and Alice will perform the correct Pauli operation.  In the noise free scenario, Alice and Bob should now share the state 
  $\ket{\phi_0^0}$, independent of the repeaters; see Fig. \ref{fig:chain}(c).
\end{enumerate}

The above describes the operations of the network, the goal of which is to establish end-to-end entanglement between Alice and Bob.  Of course, the ultimate goal of the users is to establish a shared secret key.  For this, Alice and Bob will run the entanglement based E91 protocol \cite{ekert1991quantum}.  We actually consider the more commonly used, biased version, of this protocol, where the $Z$ basis is used for key distillation and the $X$ basis is used only for testing the error rate in the channel \cite{QKD-BB84-Modification}.  In detail, Alice and Bob will perform the following operations:
\begin{enumerate}
\item Alice and Bob use the repeater chain network for $N$ rounds, each round operating as described above.  Ideally, this should result in $N$ shared Bell states held between the two users, each of the form $\ket{\phi_0^0}$.
\item Alice and Bob will choose a subset $t \subset \{1, 2, \cdots, N\}$ of size $|t| = m \le N/2$ and measure those qubits indexed by $t$ in the $X$ basis.  This results in outcomes $q_A$ (for Alice) and $q_B$ (for Bob).  They broadcast their measurement results and compute the total $X$ basis error string as $q = q_A\oplus q_B$.  This should, ideally, be the all zero string if there is no error in the network.
\item The remaining $n = N-m$ qubits held by Alice and Bob are measured in the $Z$ basis.  This will be used as their raw-key.
\item Finally, Alice and Bob run an error correcting protocol on their raw keys and a privacy amplification protocol to output their final secret key.
\end{enumerate}

\subsection{Adversarial Model and Assumptions}\label{sec:network-assumptions}
Our goal in this paper is to analyze the scenario where an adversary controls a contiguous subset (sub-network) of the repeaters in the chain, while the remaining repeaters behave honestly.  We assume that the corrupted repeaters are contiguously connected and that Eve also controls the fiber lines within this corrupted sub-network (which we call the adversary's \emph{zone of control}).  Any repeater outside the corrupted zone of control behaves honestly, and any fiber connection outside the corrupted region is not under adversarial control, but is noisy (i.e., these links are susceptible to natural noise); see Figure \ref{fig:network}.   This is in contrast to general QKD repeater chain scenarios, where it is assumed that the adversary completely controls all fiber and repeater nodes between Alice and Bob.  

As mentioned in Section \ref{sec:intro}, there are two ways to justify the above assumption.  First, in a large QKD repeater chain, it is unlikely that an adversary can gain physical access to all repeaters in the chain and all fiber links connecting them.  Instead, it is more realistic to assume an adversary can only realistically control a ``small'' subset of those repeaters and that the repeaters controlled by the adversary will be contiguous.  An alternative way to justify the assumption is that, it is likely that some repeaters near end users can be placed in secure areas (e.g., in a trusted corporate or government building).  Thus, one can justify trusting those repeaters, but not the remaining middle section of the network, connecting the two trusted regions.  Note that one does not need to know exactly how many repeaters are adversarial - instead one needs an upper-bound on this; one may just as easily assume that a certain lower-bound of repeaters are trustworthy (but noisy) and then assume the remainder are adversarial.

  Our goal is to show that improved key-rates are possible in this security model.  We will do so by deriving a bound on the finite key-rate under the network and security model derived here.  Before proceeding with our proof in the next section, however, we formally state our security model assumptions and, especially, the attack model afforded to the adversary.  The assumptions we make in our security proof are as follows:
  
  \assumption The adversary can corrupt any number of contiguous repeaters in the chain.  Furthermore, we assume that Eve can also control all fiber links between repeaters in her zone of control and the fiber links connecting to the nearest honest repeater.  We will actually assume that Eve is able to completely replace her corrupted sub-network with her own perfect devices, and perform any quantum attack possible here (i.e., she need not operate within the bounds of a repeater chain and there will be no assumption of natural noise within the corrupted sub-network).  Any repeater and fiber link outside her zone of control, however, cannot be attacked by Eve (though will be noisy).
  \assumption \label{a:cl1} The adversary can read, but not tamper with, classical messages sent by repeaters outside their zone of control.
  \assumption \label{a:cl2} Classical messages are sent after all Bell swaps are performed in the entire network for all rounds.  This can be achieved by having the network wait until all $N$ rounds have been performed, before sending the correction messages for the Bell swaps.
  \assumption \label{a:noise} Though Alice and Bob do not know exactly which sub-network is controlled by Eve, they are able to lower bound the amount of natural noise in the honest sub-network.  In particular, they are able to lower-bound the so-called \emph{noise parameter} of the network, defined in Definition \ref{def:noise-parameter}.

  Out of the above assumptions, Assumption 2 is perhaps the strongest.  We actually don't think it's entirely necessary, however could not formally prove our result without it.  We leave, as interesting future work, the removal of this assumption.  We still feel that, even with the assumption in place, our results are interesting and, furthermore, this assumption is not unreasonable in a large-scale network setting.  There may even be ways to enforce it through repeater messaging logs for instance.

\subsection{Natural Noise Model}\label{sec:naturalnoise}
Outside of the adversary's zone of control are the left and right honest repeater sub-networks.  Though honest, we will assume these are noisy in the sense that fiber noise, and internal repeater noise, may cause errors in the Bell states being distributed.  For our security proof, we will assume the natural noise acts in an i.i.d. manner and leads to a mixed Bell diagonal state.  We do not assume the noise is identical in every fiber link (e.g., some may be ``noisier'' than others).

As before, let $c$ be the total number of repeaters in the chain.  Formally, consider the link between honest repeaters $\repeater_i$ and $\repeater_{i+1}$ (if $i=0$, then we are considering the link between Alice and the first repeater $\repeater_1$, while if $i = c$, then we are considering the link between the last repeater and Bob).  Then, we assume that the two-qubit state distributed between $\repeater_i$ and $\repeater_{i+1}$ is actually of the form:
\begin{equation}\label{eq:natural-noise-single}
\rho = \sum_{x\in\B}P^i(x)\kb{\phi_x}.
\end{equation}
where $P^i(x)$ is the probability that the final shared state will be $\kb{\phi_x}$ for some $x \in B$.  Note the superscript $i$ indexes the repeater number, since we assume different links may have different noise levels.

After $N$ rounds, we can write the state between honest repeaters $i$ and $i+1$ as follows:
\begin{equation}\label{eq:natural-noise}
  \rho = \sum_{x \in \B^N}P^i(x)\kb{\phi_x}.
\end{equation}
Recall, from Section \ref{section:notation}, we define $P^i(x_1,\cdots,x_N) = P^i(x_1)P^i(x_2)\cdots P^i(x_N)$.

Our security model assumes users know something about the natural noise in the network.  To be more precise, we will assume that users can lower-bound the \emph{noise parameter} of the honest sub-network, denoted $p^*$, which is defined below:

\begin{define}\label{def:noise-parameter}
  Let $P_L^i(\ell)$ and $P_R^i(r)$, for $\ell,r \in \B$ be the probability that the $i$'th link in the left honest sub-network (respectively the right honest sub-network) produces a state $\kb{\phi_\ell}$ (respectively $\kb{\phi_r}$) on any particular single round of the network; see Equation \ref{eq:natural-noise-single}.  Let $j$ be the number of honest left sub-network links, not including the link connecting the honest network to Eve, and let $k$ be the number of honest right sub-network links, not including the link connecting to Eve (these may be zero if Eve directly connects to Alice or Bob).  For any $x \in \B$, define:
    \begin{align*}
      P_L(x) &= \sum_{\substack{\ell^1,\cdots,\ell^j\in\B\\\ell^1\oplus \cdots\oplus\ell^j = x}} P_L^1(\ell^1)\cdots P_L^j(\ell^j)\\
      P_R(x) &= \sum_{\substack{r^1,\cdots,r^k\in\B\\r^1\oplus\cdots\oplus r^k = x}} P_R^1(r^1)\cdots P_R^k(r^k)
    \end{align*}
    If $j=0$, then we simply set $P_L((0,0)) = 1$, similarly for the right network if $k=0$.
    Then, the \emph{noise parameter of the honest sub-network} is defined to be:
    \begin{equation}
      p^* = \sum_{\substack{x,y\in\B \\ x^\phase \oplus y^\phase = 1}}P_L(x)P_R(y).
    \end{equation}
  \end{define}

Note that if $j=k=0$ (i.e., there are no honest repeaters), then $p^* = 0$.

Essentially, the noise parameter characterizes the probability of there being a phase error in either the left or the right honest sub-networks, but not both.  Of course, one can always find a trivial lower-bound for this by setting $p^* = 0$, however, in this case, our key-rate expression will converge towards the normal BB84 key-rate which is expected: if $p^* = 0$, users are assuming there is no natural noise and, so, all observed noise must be the effect of an adversary system.  Once $p^* > 0$, our bound begins to improve over BB84 as we show, later, in our evaluation sections.

Finding a reasonable bound on $p^*$ will depend on context.  If we take the example of two ``safe-zones'' (as shown in Figure \ref{fig:network}), then users can characterize the link-level noise between each safe-zone repeater and use this to easily determine a value for $p^*$.  This would be a suitable lower-bound since it would assume every repeater outside the two safe-zones is adversarial (which may not actually be true and, so, in reality $p^*$ could be higher, and thus the key-rate could be higher).  Our security proof only requires a lower-bound on $p^*$ and users may be pessimistic in their choice of setting for this parameter.

  

\section{Security}

We now derive a key-rate expression for the partially corrupted repeater chain as described in the previous section.  To do this, we will first show how the system can be reduced to a three-party entanglement based protocol.  From this, we will derive a lower-bound on the quantum min entropy $\Hmin^\epsilon(A|E)$ needed to determine a bound on the number of secret key bits as per Equation \ref{eq:PA}.  Our bound will be a function of the observed $X$-basis noise in the final state shared between Alice and Bob, along with the natural noise in the honest sub-networks (or, rather, a lower-bound on the noise parameter $p^*$ defined above).

\subsection{Reduction to an Entanglement Based Protocol}

To analyze the security of a partially corrupted repeater chain, we will first show that it suffices to analyze the key-rate in the following simplified scenario where there are three honest parties: Alice, Bob, and Heidi (four parties total, counting the adversary Eve).  Here, Heidi will represent and simulate the honest sub-network in the actual network protocol (both the left and right sections); see Fig.~\ref{fig:reduction}(b).  We describe this entanglement-based version first, and then later show that analyzing this entanglement based version will lead to results for the actual protocol (described in the previous section).

\subsubsection{Entanglement Based Protocol}\label{sec:eb}
We now present the entanglement based version in its entirety.  Then, in Section \ref{sec:reduction}, we will show how this is representative of the actual network protocol discussed in Section \ref{sec:network-model}.

In this entanglement-based protocol, Heidi first creates two independent states $\ket{\mathcal{L}}_L$ and $\ket{\mathcal{R}}_R$ of the form:
\begin{align}
  \ket{\mathcal{L}} &= \sum_{\ell\in B^N}\sqrt{P_L(\ell)}\ket{\phi_{\ell}}_L\ket{F_\ell}\label{eq:left-honest}\\
  \ket{\mathcal{R}} &= \sum_{r\in B^N}\sqrt{P_R(r)}\ket{\phi_{r}}_R\ket{G_r},\label{eq:right-honest}
\end{align}
where $\braket{F_{\ell}|F_{\ell'}} = \delta_{\ell,\ell'}$ and, similarly, $\braket{G_r|G_{r'}} = \delta_{r,r'}$ where $\delta_{x,y}$ is the Kronecker Delta function.  Furthermore, we have $\ket{G_r} = \ket{G_{r_1}}\otimes\cdots\otimes\ket{G_{r_N}}$ with each $\braket{G_{r_i}|G_{r_i'}} = \delta_{r_i,r'_i}$ (and $\ket{F_\ell}$ may be written in a similar tensor form).  Note that both states are pure states.  The values of $P_L$ and $P_R$ will be determined from the actual honest network she is simulating which will be clear later (essentially, she will be simulating the network and so will set these to values based on the honest network noise, Equation \ref{eq:natural-noise}).  She keeps these states private to herself.

\begin{figure}[t]
   \centerline{\includegraphics[width=3.5in, trim = 0.cm 0.0cm 0.cm 0.1cm, clip]{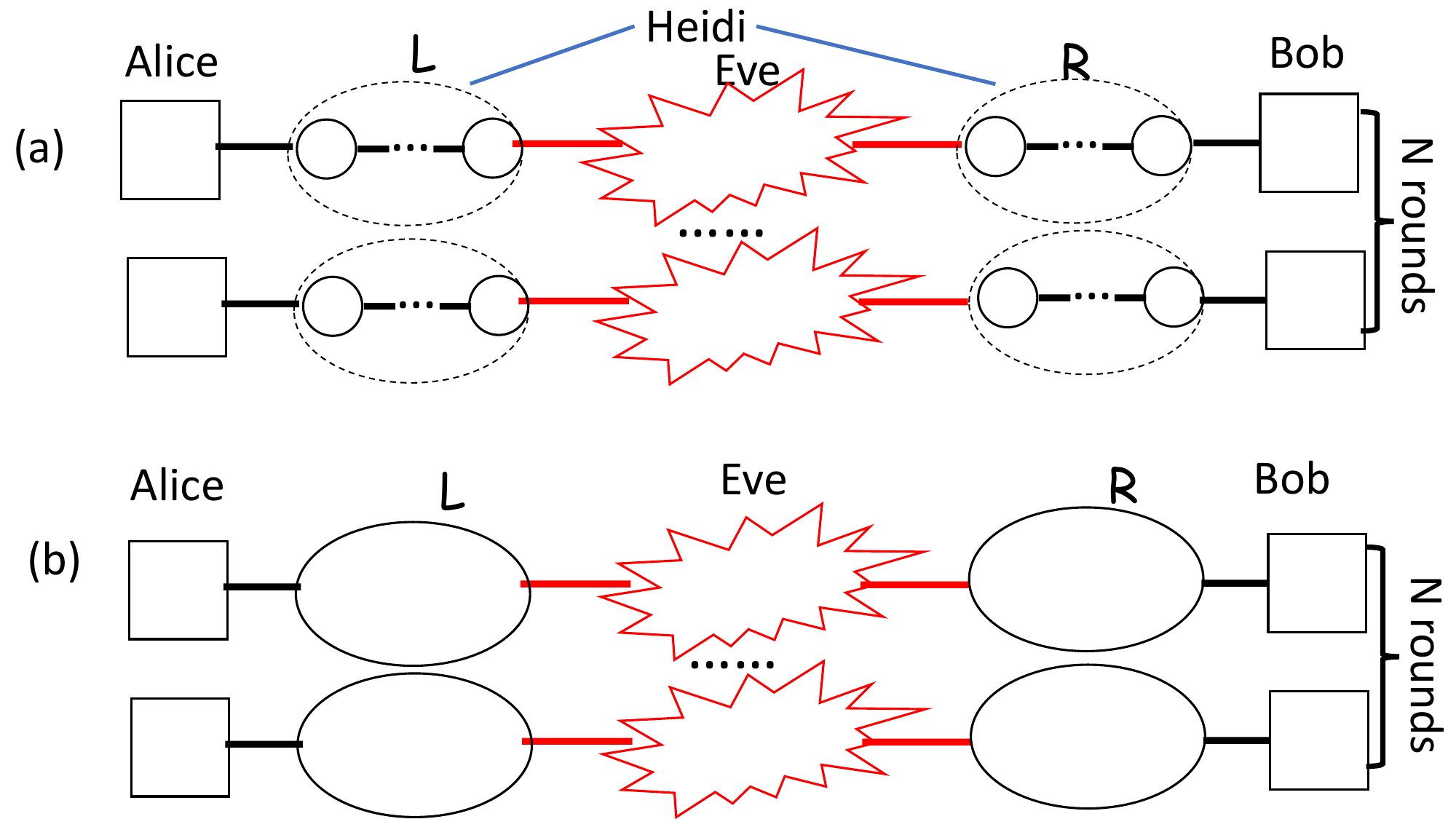}}
\caption{Illustration of the reduction: (a) the original protocol operating in $N$ rounds, and (b) the new protocol after reduction operating in $N$ rounds. Not shown here is that the original protocol (a) also requires classical communication from all repeaters under Eve's control, and each individual repeater while our reduction does not need this.  Also not shown is that Eve may receive qubits from honest repeaters in the real protocol (a), however in our reduction (b), she always prepares states and sends them to honest repeaters.}
\label{fig:reduction}
\end{figure}

Next, Eve creates an arbitrary $2N$ qubit state, entangled with her private ancilla, denoted $\ket{\widetilde{\psi}}_{ME}$, independent of the Left and Right sub-network states that Heidi created.  The $M$ (middle, since Eve is in the middle of the chain) register consists of $2N$ qubits; she then sends the  $M$ register to Heidi, while keeping her ancilla private.  After this, Heidi, who holds $6N$ qubits currently (the $L$, $M$, and $R$ registers - note that each register holds $2N$ qubits), will perform a final network operation $\net$.  This will simulate the honest network's final Bell swaps and classical messages being sent from the last honest repeaters (bordering Eve) to Alice.

This map acts on a $6N$ qubit state $\ket{\phi_\ell,\phi_i, \phi_r}$, for all $\ell,i,r\in\B^N$, in the following manner:
\begin{equation}\label{eq:net-op}
  \net^{\otimes N}\ket{\phi_\ell, \phi_i, \phi_r} = \frac{1}{2^{2N}}\sum_{z,u\in\B^N}(-1)^{g(\ell,i,r;z,u)}\ket{\mbox{``}z,u\mbox{''}}_{cl}\otimes\ket{\phi_{\ell+i + r}}_{AB},
\end{equation}
for some function $g:\B^{5N}\rightarrow \{0,1\}$ whose exact action, though easy to actually derive by simulating the action of the final left and right sub-network Bell swaps (which we do later), is not important to our current discussion. This map $\net$ will simulate the honest network's final operations of performing Bell swaps at the ``border'' repeaters (those nearest Eve) and sending the measurement outcomes.  It will also simulate the final Pauli fix applied by Alice.  Note that the input of $g$ depends on $5N$ Bell states since this is a function of all $N$ rounds and, for each round, there are five important Bell values: the two messages output by the map (output by the honest network or Heidi in this case), and the three input Bell states (from $M$, $L$, and $R$).  Also, recall the additive notation for Bell states, defined in Section \ref{section:notation}.

The above map is actually unitary, which is not obvious from the above definition since $g$'s action is not defined; however, that it is an unitary, will be clear when we write out explicitly how the map operates in the next section.  Essentially, it performs a SWAP followed by a delayed Bell basis measurement and recording the result in the ``$cl$'' register; it then performs the final Pauli correction which Alice would normally have done.  We use quotes in the ``$cl$'' register as they will, later, represent the classical message sent from the last two repeaters to Alice.

Finally the $A$ and $B$ registers (which are $N$-qubits each) are sent to Alice and Bob respectively, while the ``$cl$'' register (the classical message register) is measured and the outcome broadcasted to all parties - this represents the final correction term that normally would have been broadcast (i.e., it represents the sum of all the honest messages that normally would have been broadcast, not the entire message transcript that would have been broadcast in the actual prepare and measure protocol).  The $G$ and $F$ registers are simply discarded (i.e., traced out).  Alice and Bob are then free to run the E91 protocol as normal, namely, they choose a random sample and measure in the $X$ basis to test the fidelity of the state, while measuring the remaining systems in the $Z$ basis to derive a secret key after error correction and privacy amplification.  See Figure \ref{fig:eb-protocol}.

\begin{figure}
   \centerline{\includegraphics[width=.75\linewidth, trim = 1.5cm 20.0cm 1.5cm 1.1cm, clip]{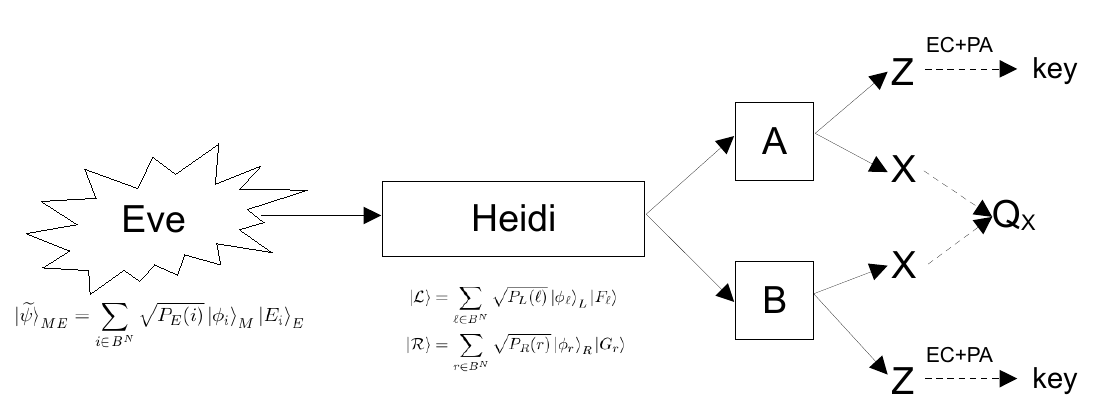}}
\caption{A high-level view of the entanglement-based protocol.  First, Eve creates a state $\ket{\widetilde{\psi}}_{ME}$, and sends the $M$ portion to Heidi.  Heidi, simulating the honest portion of the network, creates Left and Right network states, then applies the network operator $\net$.  She then sends the resulting qubits to Alice and Bob.  Alice and Bob then measure some of the qubits in the $X$ basis, used to compute $Q_X$, the parity of their $X$ basis outcomes.  The remaining qubits are measured in the $Z$ basis which is used to produce their final secret key after Error Correction (EC) and Privacy Amplification (PA).  We show that security of this entanglement based protocol will imply security of the actual repeater-chain protocol.  Not shown in this figure is the classical communication sent from Heidi, and also from Alice and Bob.}
\label{fig:eb-protocol}
\end{figure}

\subsubsection{Reduction}\label{sec:reduction}

We claim security of the above entanglement-based protocol will imply security of the actual partially corrupted repeater chain.  To show this, we trace the execution of the actual network, and compare with the state that would have been produced in the entanglement based protocol above.

First, since messages are not sent until after all $N$ rounds are complete and, furthermore, since Eve cannot tamper with these messages (see Assumptions 2 and 3 as defined in Section \ref{sec:network-assumptions}), Eve cannot adapt her attack based on the results of the honest network's Bell messages.  On account of this, giving Eve the ability to prepare all $2N$ qubits simultaneously for those repeaters neighboring her can only give her a greater advantage than in the potentially more realistic case where (1) she would need to feed in qubits to the neighboring repeaters in smaller blocks and/or (2) she would receive qubits prepared by the honest network instead of being free to create and send her own.  Thus, we analyze the case where Eve prepares all $2N$ qubits for the neighboring repeaters as this can only be to her advantage - any other scenario would give Eve fewer attack opportunities and, thus, more uncertainty in her system.  We can thus assume that Eve prepares the state:
\begin{equation}
  \ket{\widetilde{\psi}}_{ME} = \sum_{i\in\B^N}\sqrt{P_E(i)}\ket{\phi_i}_M\ket{E_i}_E,
\end{equation}
where each $\ket{E_i}$ is some arbitrary normalized, but not necessarily orthogonal, state in Eve's ancilla, and $P_E(i)$ is arbitrary such that $\sum_iP_E(i) = 1$ (we are not assuming an i.i.d. attack).  Note that Eve is allowed, in the actual network scenario, to send a classical message transcript (consisting of $2N$-bits for every repeater under her control).  However, this message can only cause Alice to apply an incorrect Pauli gate later when correcting the Bell state.  Eve can simulate this, without sending a classical message, by applying a suitable Pauli gate herself to her state above, before sending the qubits to the corresponding neighboring honest repeaters.  Thus, Eve's ability to send a classical message (the output of her adversarial repeater network), does not give her any additional power if we assume she is allowed to always prepare qubit states and send them to her neighbors and so we do not need to consider it further.  We do, however, need to consider the classical messages sent by the honest repeaters as that message transcript may be affected by Eve's state above as we soon see.

Now, let's consider the action of the honest network.  First, it attempts to establish link-level entanglement on all edges between honest repeaters (but not between honest repeaters and Eve's corrupted region as those links are under Eve's control as discussed) - see Figure \ref{fig:reduction}.  By Assumption 4, the noise in these links results in a mixed Bell state.  Let $j$ be the number of fiber links in the left honest network (not counting the fiber link leading to Eve which, we assume, is corrupted) and $k$ be the number of fiber links in the right honest network.  Let $\rho^i_L$ (respectively $\rho^i_R$) be the state of the quantum system shared between nodes on link $i$ on the left (respectively right) sub-network.  By Assumption 4 and Equation \ref{eq:natural-noise}, we can write these states as:
\begin{align}
  \rho_L^i = \sum_{\ell^i\in\B^N}P^i_L(\ell^i)\kb{\phi_{\ell^i}}_L\\
  \rho_R^i = \sum_{r^i\in\B^N}P^i_R(r^i)\kb{\phi_{r^i}}_R
\end{align}
At this point, the joint state of the network is:
\[
  \rho_L^1\otimes\cdots\rho_L^j\otimes\kb{\widetilde{\psi}}\otimes\rho_R^1\otimes\cdots\otimes\rho_R^k.
\]

The honest network will now perform Bell operations.  Note that the actual order of the Bell measurements performed by the network does not matter, since all Bell swaps at this point are performed by honest repeaters.  In fact, the network might as well perform Bell swaps on the left and right honest networks while waiting for Eve's $2N$ qubits.

Let's look closer at the Bell swap operation performed by a single honest repeater.  Assume that we have a chain: $X-Y-Z$, where $X-Y$ and $Y-Z$ share some quantum state of the form $\ket{\phi_a}$ and $\ket{\phi_b}$ for some $a,b\in\B$.  Assume repeater $Y$ is performing the Bell swap.  Note that:
\begin{align*}
  \ket{\phi_a,\phi_b} &= \frac{1}{2}\left(\ket{0,a^\bit,0,b^\bit} + (-1)^{b^\phase}\ket{0,a^\bit, 1, \bar{b}^\bit} + (-1)^{a^\phase}\ket{1,\bar{a}^\bit, 0, b^\bit} + (-1)^{a^\phase+b^\phase}\ket{1, \bar{a}^\bit, 1, \bar{b}^\bit}\right)\\
  &\cong\frac{1}{2}\left(\ket{a^\bit,0}\ket{0,b^\bit} + (-1)^{b^\phase}\ket{a^\bit,1}\ket{0,\bar{b}^\bit} + (-1)^{a^\phase}\ket{\bar{a}^\bit,0}\ket{1,b^\bit} + (-1)^{a^\phase+b^\phase}\ket{\bar{a}^\bit, 1}\ket{1, \bar{b}^\bit}\right),
\end{align*}
where, above, we simply permuted the subspaces so that the two qubits owned by $Y$ are on the left while the third qubit is the one held by $X$ and the fourth (right-most) qubit is the one held by $Z$.  Repeater $Y$ is now going to perform a Bell measurement on the qubits it holds (the left-two qubits in the above, permuted, state).

If $a^\bit = 0$, the above state is found to be:
\begin{align*}
  \ket{\phi_a,\phi_b} \cong &\frac{1}{2\sqrt{2}}\ket{\phi_0^0}\otimes \left(\ket{0,b^\bit} + (-1)^{a^\phase + b^\phase}\ket{1, \bar{b}^\bit}\right)\\
                      +&\frac{1}{2\sqrt{2}}\ket{\phi_0^1}\otimes \left(\ket{0,b^\bit} + (-1)^{a^\phase + b^\phase + 1}\ket{1, \bar{b}^\bit}\right)\\
                      +(-1)^{b^\phase}&\frac{1}{2\sqrt{2}}\ket{\phi_1^0}\otimes \left(\ket{0,\bar{b}^\bit} + (-1)^{a^\phase+b^\phase}\ket{1, b^\bit}\right)\\
  +(-1)^{b^\phase}&\frac{1}{2\sqrt{2}}\ket{\phi_1^1}\otimes \left(\ket{0,\bar{b}^\bit} + (-1)^{a^\phase+b^\phase+1}\ket{1, b^\bit}\right)
\end{align*}
Otherwise, if $a^\bit = 1$, then the state is:
\begin{align*}
  \ket{\phi_a,\phi_b} \cong (-1)^{b^\phase}&\frac{1}{2\sqrt{2}}\ket{\phi_0^0}\otimes \left(\ket{0,\bar{b}^\bit} + (-1)^{a^\phase+b^\phase}\ket{1, b^\bit}\right)\\
                      +(-1)^{b^\phase+1}&\frac{1}{2\sqrt{2}}\ket{\phi_0^1}\otimes \left(\ket{0,\bar{b}^\bit} + (-1)^{a^\phase+b^\phase+1}\ket{1, b^\bit}\right)\\
                      +&\frac{1}{2\sqrt{2}}\ket{\phi_1^0}\otimes \left(\ket{0,b^\bit} + (-1)^{a^\phase+b^\phase}\ket{1, \bar{b}^\bit}\right)\\
  +&\frac{1}{2\sqrt{2}}\ket{\phi_1^1}\otimes \left(\ket{0,b^\bit} + (-1)^{a^\phase+b^\phase+1}\ket{1, \bar{b}^\bit}\right)
\end{align*}
Thus, for any $a^\bit$, we can write the joint state as:
\begin{equation}\label{eq:sec:bell-action}
  \ket{\phi_a,\phi_b}_{XYZ} \cong \frac{1}{2}\sum_{x\in\B}(-1)^{f(a,b,x)}\ket{\phi_x}_Y\otimes\ket{\phi_{a+b+x}}_{XZ},
\end{equation}
for some function $f:\B^3\rightarrow \{0,1\}$ which can be determined from the above, though the exact mapping is not important at the moment.  Essentially, the above is defining a map that applies a SWAP to the middle two qubits which will simulate a delayed measurement of the middle repeater $Y$ (the left-most register will later be measured in the Bell basis, and the outcome reported).

The left two registers are then measured in the Bell basis leading to a classical outcome $x\in\B$ which is saved in a classical register spanned by $\ket{\mbox{``}x\mbox{''}}$ for all $x\in\B$.  Note we use quotes, here, to distinguish the fact that this system is classical and will later be used as a message system.  The final state, then, is:
\begin{equation}
  \ket{\phi_a,\phi_b} \mapsto \frac{1}{4}\sum_{x\in\B}\kb{\mbox{``}x\mbox{''}}\otimes\kb{\phi_{a+b+x}},
\end{equation}
where, recall from Section \ref{section:notation}, $\ket{\phi_{a+b+x}} = \ket{\phi_{a^\bit+b^\bit+x^\bit}^{a^\phase + b^\phase + x^\phase}}$ (where all arithmetic in the subscript and superscript is done bitwise modulo two of course).

Returning to the protocol, let's assume the left and right honest sub-networks perform their Bell swaps $N$ times on each repeater, except for the repeaters neighboring Eve.  On the left sub-network, the state becomes, after all $j-1$ repeaters perform the above Bell measurements on all $N$ rounds:
\begin{align}
  \rho_L^1\otimes\cdots\rho_L^j &\mapsto \frac{1}{4^{N(j-1)}}\sum_{x^1,\cdots,x^{j-1}\in\B^N}\kb{\mbox{``}{x^1, \cdots, x^{j-1}}\mbox{''}}\otimes\sum_{\ell^1,\cdots, \ell^j\in\B^N} P_L^1(\ell^1)\cdots P_L^j(\ell^j)\kb{\phi_{\ell^1+\cdots+\ell^j+x^1+\cdots+x^{j-1}}}\notag\\
  &= \frac{1}{4^{N(j-1)}}\sum_{x^1,\cdots,x^{j-1}\in\B^N}\kb{\mbox{``}{x^1, \cdots, x^{j-1}}\mbox{''}}\otimes\sum_{\ell\in\B^N}P_L(\ell)\kb{\phi_{\ell+x^1+\cdots+x^{j-1}}} := \rho_L,
\end{align}
where:
\begin{equation}
  P_L(\ell) = \sum_{\substack{\ell^1,\cdots,\ell^j\in\B^N\\\ell^1+\cdots+\ell^j = \ell}}P_L^1(\ell^1)\cdots P_L^j(\ell^j).
\end{equation}
Similarly, the right honest sub-network's state can be written as:
\begin{equation}
  \rho_R = \frac{1}{4^{N(k-1)}}\sum_{y^1,\cdots,y^{k-1}\in\B^N}\kb{\mbox{``}{y^1, \cdots, y^{k-1}}\mbox{''}}\otimes\sum_{r\in\B^N}P_R(r)\kb{\phi_{r+y^1+\cdots+y^{k-1}}},
\end{equation}
with a similar definition for $P_R(r)$ based on each $P_R^i(\cdot)$.  The overall network, then, is in the state:
\[
  \rho_L\otimes\kb{\widetilde{\psi}}\otimes\rho_R.
\]

The two last repeaters (one to the left of Eve and one to the right) now perform their final Bell swaps.  To model this, let's consider a purification of the above network state:
\begin{align}
  \rho_L\otimes\kb{\widetilde{\psi}}\otimes\rho_R = &tr_{GF}\left(\frac{1}{4^{N(j+k-2)}}\sum_{\vec{x},\vec{y}}\ket{\mbox{``}{\vec{x},\vec{y}}\mbox{''}}\bra{\mbox{``}{\vec{x},\vec{y}}\mbox{''}}\right.\notag\\
  &\otimes \left.P\left(\sum_{\ell,i,r\in\B^N}\sqrt{P_L(\ell)P_R(r)P_E(i)}\ket{\phi_{\ell+\vec{x}}, \phi_i,\phi_{r+\vec{y}}}\ket{F_\ell, G_r, E_i}\right)\right)
\end{align}
Above, the sum is over vectors of messages $\vec{x} = (x^1,\cdots, x^{j-1})$ and $\vec{y} = (y^1,\cdots, y^{k-1})$.
Also, we purified the left and right network states independently by appending ancillas spanned by $\ket{F_\ell}$ (for the left honest network) and $\ket{G_r}$ for the right network.  Of course $\braket{F_\ell|F_{\ell'}} = \delta_{\ell,\ell'}$ and similarly for $\braket{G_r|G_{r'}}$.  Furthermore, we may write $\ket{F_\ell} = \ket{F_{\ell_1}}\ket{F_{\ell_2}}\cdots$ and, similarly, for $\ket{G_r}$ due to the i.i.d. structure of the natural noise..

Consider a particular $\vec{x}$ and $\vec{y}$.  Using the SWAP and delayed measurement technique, used to derive Equation \ref{eq:sec:bell-action}, twice, first for the left honest repeater which owns half the qubits of $\phi_{\ell+\vec{x}}$ and half the qubits from $\phi_i$, then again for the right honest repeater which owns the other half of $\phi_i$ and half of $\phi_{r+\vec{y}}$, yields the state, for this particular outcome $\vec{x}$ and $\vec{y}$:
\begin{align}
  &\sum_{\ell,i,r\in\B^N}\sqrt{P_L(\ell)P_R(r)P_E(i)}\ket{\phi_{\ell+\vec{x}}, \phi_i,\phi_{r+\vec{y}}}\ket{F_\ell, G_r, E_i}\notag\\
  \mapsto&\frac{1}{2^{2N}}\sum_{z,u\in\B^N}\ket{\phi_z,\phi_u}_{ZU}\otimes\sum_{\ell,i,r\in\B^N}(-1)^{g(\ell,i,r;z,u)}\sqrt{P_L(\ell)P_R(r)P_E(i)}\ket{\phi_{\ell+i+r+\vec{x}+\vec{y} + z + u}}\ket{F_\ell, G_r, E_i}\label{eq:net-op-real}
\end{align}
where the function $g:\B^{5N}\rightarrow\{0,1\}$ can be determined through the action of function $f$, though, again, its exact structure is not important to the analysis.  Note this is the same function $g$ as in Equation \ref{eq:net-op}.  The repeaters will then measure the $ZU$ system resulting in:
\begin{equation}
  \frac{1}{4^{2N}}\sum_{z,u\in\B^N}\kb{\mbox{``}z,u\mbox{''}}\otimes P\left(\sum_{\ell,i,r\in\B^N}(-1)^{g(\ell,i,r;z,u)}\sqrt{P_L(\ell)P_R(r)P_E(i)}\ket{\phi_{\ell+i+r+\vec{x}+\vec{y} + z + u}}\ket{F_\ell, G_r, E_i}\right)
\end{equation}
Note that, the phase induced by the $g$ function will affect the probability of a repeater observing a particular message $z,u$ - however, working out algebraically exactly this distribution turns out to be not necessary for the min-entropy analysis later; thus, while the $g$ function is important, writing out its exact action is not, for the sake of our security proof.

From the above, the message transcript $\vec{x}$, $\vec{y}$, $z$, and $u$ are sent to Alice who performs the correct Pauli operations.  The ordering of these is not important - let's assume she fixes the $\vec{x}$ and $\vec{y}$ portions first, resulting in state:
\begin{equation}
  \sigma = \frac{1}{4^{2N}}\sum_{z,u\in\B^N}\kb{\mbox{``}z,u\mbox{''}}\otimes P\left(\sum_{\ell,i,r\in\B^N}(-1)^{g(\ell,i,r;z,u)}\sqrt{P_L(\ell)P_R(r)P_E(i)}\ket{\phi_{\ell+i+r + z + u}}\ket{F_\ell, G_r, E_i}\right)
\end{equation}
Note that the above was all conditioning on a particular, but arbitrary, outcome $\vec{x}$ and $\vec{y}$.  Considering, now, all possible $\vec{x}$ and $\vec{y}$ as a mixed state, before Alice applies operators to fix the $z$ and $u$ messages, the entire system, then, can be written as:
\begin{equation}\label{eq:final-state-real}
  tr_{GF}\left(\frac{1}{4^{N(j+k-2)}}\sum_{\vec{x},\vec{y}}\kbb{\mbox{``}{\vec{x},\vec{y}}\mbox{''}}\otimes \sigma\right)
\end{equation}
Note that the $\vec{x}$ and $\vec{y}$ message systems are now independent of the state $\sigma$.  Thus, regardless of the message outcome of the honest repeater network (the portion that did not interact with the adversary), the overall entropy computation will be identical. This is not too surprising considering the operation of an honest repeater with honest inputs, and the fact that the noise in the left and right honest sub-networks is well characterized.  Therefore, there is no need to consider this part of the message transcript when analyzing the security of the protocol.  Note that it \emph{is} important to consider the message transcript from the honest repeaters that do interact with Eve (the $Z$ and $U$ messages), as these messages may not be independent of the final state, especially Eve's ancilla state.

Equation \ref{eq:final-state-real} represents the system before Alice applies the Pauli gates based on the $z$ and $u$ messages (from the repeaters bordering Eve).  Since the $\vec{x}$ and $\vec{y}$ messages are independent, the only important element is the state $\sigma$.  That is, if we know Eve's uncertainty in $\sigma$, we can derive the key-rate of the system.  However, it is not difficult to see that the entanglement based protocol, described earlier, will produce exactly the same state $\sigma$.  Indeed, by tracing the protocol's execution, where Eve produces the same attack state in the entanglement based version, and Heidi uses the honest network probabilities $P_L$ and $P_R$ for her created states, the resulting system in the entanglement based version will produce, exactly, $\sigma$.  Note that the network operation map $\net^{\otimes N}$ (Equation \ref{eq:net-op}) is defined to simulate, exactly, the Bell swap operations of the honest network, derived above (in particular, to simulate Equation \ref{eq:net-op-real}).  Thus, we can conclude that it suffices to analyze the entanglement based version - any entropy computation there will translate to an equivalent computation in the real network setting (Equation \ref{eq:final-state-real}).

\subsection{Key-Rate Bound}

We now derive a key-rate bound for the partially corrupted repeater chain, or, rather, the entanglement based protocol discussed in the previous section.  To do so, we first derive a bound on the conditional min entropy after running the E91 protocol in this network setup.  Equation \ref{eq:PA} can then be used to translate this to a key-rate bound.  Our main result is stated, and proven, in the following theorem:

\begin{theorem}\label{thm:main}
  Let $\epsilon > 0$ and let $\rho_{ABE}$ be the state produced through the entanglement based protocol discussed above in Section \ref{sec:eb}, where the $A$ and $B$ registers consist of $N$ qubits each. Let $p^*$ be the noise parameter of the honest network (or a lower-bound on the noise parameter) as defined in Definition \ref{def:noise-parameter}.  Assume that Alice and Bob choose a random subset $t$ of size $m\le N/2$ and measure their corresponding qubits in the $X$ basis resulting in outcome $q_A$ and $q_B$ (which are $m$-bit strings).  Denote by $Q_X$ to be $Q_X = q_A\oplus q_B$ (this represents the $X$-basis error string).  Alice and Bob then measure their remaining $n=N-m$ qubits in the $Z$ basis resulting in state $\rho^{(t,q_A,q_B)}_{ABE}$. Then it holds that:
  \begin{equation}
    Pr\left[\Hmin^{8\epsilon + 2(2\epsilon)^{1/3}}(A|E)_{\rho^{(t,q_A,q_B)}} \ge n\left(1 - \bar{h}\left(\frac{w\left(Q_X\right) - p^* + \delta'}{1-2p^*} + \delta\right)\right)\right] \ge 1 - 2(2\epsilon)^{1/3},
  \end{equation}
  where the probability is over the subset choice $t$ and the observed outcomes $q_A,q_B$, and where:
  \begin{align}
    \delta = \sqrt{\frac{(N+2)\ln\frac{2}{\epsilon^2}}{mN}} &&\text{and}&& \delta' = \sqrt{\frac{\ln\frac{2}{\epsilon}}{2m}}.\label{eq:thm:deltap}  
  \end{align}
\end{theorem}
  \textbf{Sketch of Proof: }
  Our proof proceeds in three main steps.  First, we will use Theorem \ref{thm:sample} to construct ideal states and model the protocol under these states.  Next, we will define a new ``ideal-ideal'' state where the previous ideal states are better behaved in a noisy network.  Finally, we analyze the entropy of this ideal-ideal state, and then use Lemma \ref{lemma:cptp-entropy} to promote the analysis to the real network state.  Note that steps one and three follow from ideas used in the proof of sampling-based entropic uncertainty relations \cite{yao2022quantum,krawec2019quantum}.

  In some more detail, consider the overall ``real'' state of the system: Eve creates $\ket{\widetilde{\psi}}_{ME}$, giving the $2N$ qubit $M$ register to the honest network simulator Heidi, while Heidi also creates $\ket{\mathcal{L}} \otimes \ket{\mathcal{R}}$.  The state Heidi creates is well known due to our assumptions on the honest network noise; the state that Eve creates is unknown.  Heidi then operates on the $L, M,$ and $R$ registers using operator $\net^{\otimes N}$ (see Equation \ref{eq:net-op}).  This produces a quantum state $\rho_{ABE}$.  Alice and Bob then choose a random sample $t$ of size $m < N/2$, measure their qubits, indexed by this subset, in the $X$ basis to test the fidelity of the channel, reporting the parity of their outcomes $Q_X\in \{0,1\}^m$ (which should be the all zero string - any ``one'' in this string indicates an $X$ basis error).  The difficulty is that, given a particular $Q_X$ and subset $t$, we need to know how many errors Eve's arbitrary state induced.

  Given Eve's state, we can construct ideal states using Theorem \ref{thm:sample} which are, on average over the subset choice, $\epsilon$-close to the real state $\ket{\widetilde{\psi}}$.  Heidi will create the same $\ket{\mathcal{L}}$ and $\ket{\mathcal{R}}$ states as before and perform the same network operation, but now on the joint state consisting of an ideal state $\ket{\nu^t}$.  Since quantum operations cannot increase trace distance, the resulting state is $\epsilon$-close to the actual state we want to analyze.

  Even now, however, analyzing the quantum min entropy of this state cannot easily be done.  Instead, we must argue that, with high probability, the natural phase errors cannot always cancel out Eve's induced error.  Since quantum min entropy is a ``worst-case'' entropy, analyzing the entropy of the above ``ideal'' state will result in a trivial, or at least very poor, bound.  Instead, we must next define what we call ``ideal-ideal'' states - ideal states which have a further ideal property that the natural noise actually adds to Eve's overall induced error, instead of taking away.  Furthermore, these ideal-ideal states have the additional property that, given the observed phase error and the known network noise (Definition \ref{def:noise-parameter}), we can upper-bound Eve's phase error, even though we cannot directly observe it.  These new states are $2\epsilon$-close to the actual real state and, so, arguments made about the min-entropy of these ideal-ideal states will translate, using Lemma \ref{lemma:cptp-entropy}, to the real state, giving us our result.
\begin{proof}
  We now formally prove Theorem \ref{thm:main}.
  
  \textbf{Step 1 - Ideal State Construction: }
Consider the entanglement based protocol discussed in Section \ref{sec:eb}.  Let $\ket{\widetilde{\psi}}_{ME}$ be the state Eve creates, where the $M$ register consists of $2N$-qubits and the $E$ system is arbitrary. The $M$ system is sent to Heidi.  As discussed in the previous section, Heidi creates two purified states of the form:
\begin{equation}
  \ket{\mathcal{L}} = \sum_{\ell \in \B^N}\sqrt{P_L(\ell)}\ket{\phi_{\ell^\bit}^{\ell^\phase}}_L\ket{F_\ell}
\end{equation}
and:
\begin{equation}
  \ket{\mathcal{R}} = \sum_{r \in \B^N}\sqrt{P_R(r)}\ket{\phi_{r^\bit}^{r^\phase}}_R\ket{G_r},
\end{equation}
where $P_L(\ell)$ and $P_R(r)$ depend on the honest left and right sub-networks that Heidi is simulating (see the previous section and, in particular, Equations \ref{eq:left-honest} and \ref{eq:right-honest}). Note the $\ket{F_\ell}$ states are orthonormal and separable (similar for the $\ket{G_r}$ states).  The joint state, therefore, may be written (up to a permutation of the subspaces) as:
\begin{equation}\label{eq:real-state}
  \sum_{\ell, r \in \B^m} \sqrt{P_L(\ell)P_R(r)}\ket{\phi_{\ell^\bit}^{\ell^\phase}}\ket{\phi_{r^\bit}^{r^\phase}}\ket{\widetilde{\psi}}_{ME}\ket{F_\ell}\ket{G_r}
\end{equation}

From the above state, Heidi will perform the network operation $\net^{\otimes N}$ (see Equation \ref{eq:net-op}), which will simulate the honest sub-networks' Bell swaps, message transmissions, and Alice's final Pauli correction; Heidi then sends $N$ qubits to Alice, $N$ qubits to Bob, and broadcasts the final correction term (a $2N$ bit message).  Alice and Bob will then run the E91 protocol on the resulting state.  Namely, they will choose a random subset $t$ of size $m$, measure their qubits indexed by this subset in the $X$ basis, compute and report the parity of the outcomes $Q_X$ (which, ideally, should be the zero string), and then measure the remaining qubits in the $Z$ basis.  These last measurements form the raw key which is subsequently processed further using error correction and privacy amplification.  Call this final state $\rho^t_{ABE}$ (which is a mixed state over all possible $Q_X$ observations).

Instead of analyzing the actual state of the system above, we will instead apply Theorem \ref{thm:sample} and analyze ideal states where sampling is well-behaved.  We apply the theorem not to the entire state, Equation \ref{eq:real-state}, but instead to the state Eve prepares $\ket{\widetilde{\psi}}_{ME}$.  We use the sampling strategy $\Psi_4$, analyzed in Lemma \ref{lemma:sample-used}, with respect to the Bell basis. From this, we have ideal states $\{\ket{\ideal^t}_{ME}\}_t$, indexed by subsets $t$, such that:
\begin{equation}\label{eq:td}
  \frac{1}{2}\trd{\sum_tP_T(t)\kb{t}\otimes\left(\kb{\widetilde{\psi}}_{ME} - \kb{\ideal^t}_{ME}\right)} \le \sqrt{\epsilon^{cl}} = \epsilon,
\end{equation}
where the last equality follows from Lemma \ref{lemma:sample-used} and our choice of $\delta$.
Above, each $\ket{\ideal^t}_{ME} \in \text{span}\left(\mathcal{G}_t\right)\otimes\mathcal{H}_E$, with:
\begin{equation}
  \mathcal{G}_t = \left\{i\in\B^N \st w(i^\phase_t) \dc w(i_{-t}^\phase)\right\}
\end{equation}
Note that we are applying Theorem \ref{thm:sample} just to the state Eve creates and not to the entire joint state.  We now compute the min entropy in this ideal case and will later promote this analysis to the real case (where Eve's states are not ideal).

Consider a particular subset $t$ and ideal state $\ket{\ideal^t}$.  It is clear that we may write this state in the following form (up to a permutation on individual qubit subspaces):
\begin{equation}
    \ket{\ideal^t} \cong \sum_{i_t\in\B^m}\sqrt{P_E(i_t)}\ket{\phi_{i_t}} \otimes \sum_{i_{-t}\in J_{i_t}}\beta_{i_{-t}|i_t}\ket{\phi_{i_{-t}}}\ket{E_{i_{-t}|i_t}}.
\end{equation}
where $P_E(i_t)$ is a value determined by Eve, which Alice and Bob cannot directly observe, and where $J_{x} = \{y\in\B^n \st w(y^\phase) \dc w(x^\phase)\}$.

Tensoring this with the honest sub-network states (which we do not apply Theorem \ref{thm:sample} to, but instead analyze directly) yields the following state (up to a permutation of subspaces):
\begin{align}
  &\sum_{\ell_t, r_t, i_t\in B^m}\sqrt{P_L(\ell_t)P_R(r_t)P_E(i_T)}\ket{\phi_{\ell_t}, \phi_{i_t}, \phi_{r_t}}\ket{F_{\ell_t},G_{r_t}}\notag\\
  &\underbrace{\otimes\sum_{\ell_{-t},r_{-t}\in\B^n}\sum_{i_{-t}\in J_{i_t}}\sqrt{P_L(\ell_t)P_R(r_t)}\beta_{i_{-t}|i_t}\ket{\phi_{\ell_{-t}}, \phi_{i_{-t}}, \phi_{r_{-t}}}\ket{E_{i_{-t}|i_t}}\ket{F_{\ell_{-t}},G_{r_{-t}}}}_{\mu(\ell_t,r_t,i_t)}\notag\\\notag\\
  & = \sum_{\ell_t, r_t, i_t\in B^m}\sqrt{P_L(\ell_t)P_R(r_t)P_E(i_T)}\ket{\phi_{\ell_t}, \phi_{i_t}, \phi_{r_t}}\ket{F_{\ell_t},G_{r_t}} \otimes \ket{\mu(\ell_t,r_t,i_t)}.\label{eq:mu}
\end{align}
Note we are only permuting the subspaces for clarity in presentation and to make the algebra simpler; in practice, the parties do not need to do this, and it does not affect the final min entropy bound if users do or do not perform this permutation of subspaces.

Now, the final network operation is performed $\net^{\otimes N}$ (Equation \ref{eq:net-op}) on the above state (Equation \ref{eq:mu}).  However, it is equivalent to analyze the case where the network operates on the subset $t$ first, Alice and Bob measure in the $X$ basis in the $t$ subset only, and then the network operates on the complement $-t$ while Alice and Bob then measure, in the $Z$ basis, the complement qubits $-t$.

After the network operation on $t$, the state becomes (again, writing the state after permuting subspaces for clarity):
\begin{align}
  \frac{1}{2^{2m}}&\sum_{\ell_t,r_t,i_t\in\B^m}\sqrt{P_L(\ell_t)P_R(r_t)P_E(i_t)}\ket{F_{\ell_t}, G_{r_t}}\notag\\
  &\otimes\sum_{z,u\in\B^m}(-1)^{g(\ell_t,i_t,r_t;z,u)}\ket{\mbox{``}z, u\mbox{''}}_{cl}\ket{\phi_{\ell_t + i_t + r_t}}_{AB}\otimes \ket{\mu(\ell_t,r_t,i_t)}.\label{eq:ideal-network-act}
\end{align}
(Recall, again, the additive notation for Bell states defined in Section \ref{section:notation}.)

Now, an $X$ basis measurement is made by Alice and Bob on those qubits indexed by $t$ and the parity is reported (in practice, Alice and Bob broadcast their measurement results and compute the parity, however, ultimately, the parity is the important information).  Let $X_0$ and $X_1$ be two POVM elements where $X_0 = \kb{+,+}_{AB} + \kb{-,-}_{AB}$ and $X_1 = \kb{+,-}_{AB} + \kb{-,+}_{AB}$.  That is, $X_0$ indicates a parity of zero in Alice and Bob's measurements, i.e., Alice and Bob receive the same measurement outcome, while $X_1$ indicates an error in Alice and Bob's $X$ basis measurement (i.e., they get opposite measurement outcomes in that basis).  The following is easy to verify:

\begin{align*}
  X_j\ket{\phi_x^y} &= \left\{\begin{array}{ll}
                              \ket{\phi_x^y} & \text{ if } y = j\\
                              0 & \text{ otherwise}
                            \end{array} \right.
\end{align*}

Finally, let $\widetilde{X}$ be the CPTP map which measures Alice's and Bob's qubits indexed by $t$, records the parity result in a separate register (which we will call the $P$ register), and finally traces out the post measured quantum state.  Applying this map $\widetilde{X}$ to the ideal state above (Equation \ref{eq:ideal-network-act}), and also tracing out the $G$ and $F$ systems (or, rather, the $t$ portion of the $G$ and $F$ systems), yields:
\begin{align}
  \sigma^t=&\frac{1}{4^{2m}}\sum_{z,u\in\B^m}\kb{\mbox{``}z,u\mbox{''}}_{cl}\otimes\left(\sum_{\ell, r, i\in\B^m}P_L(\ell)P_R(r)P_E(i)\kb{l^\phase \oplus r^\phase \oplus i^\phase}_P\kb{\mu(\ell, r, t)}\right) \nonumber \\
  =&\frac{1}{4^{2m}}\sum_{z,u\in\B^m}\kb{\mbox{``}z,u\mbox{''}}_{cl}\otimes\left(\sum_{i\in\B^m}P_E(i)\left[\sum_{\ell, r\in\B^m}P_L(\ell)P_R(r)\kb{l^\phase \oplus r^\phase \oplus i^\phase}_P\kb{\mu(\ell, r, t)}\right]\right) 
\end{align}
Note that, above, we have removed the $t$ subscript in the notation for clarity, since the context is clear.  Also, observe that the above state is a mixture of states which are not necessarily normalized; that is, the probability of observing any particular $z$ or $u$ is not necessarily uniform due to the influence of the $P_E(i)$ term.

The above is the ideal-state using states constructed from Theorem \ref{thm:sample}.  The real state $\rho^t_{ABE}$ can be found similarly to the above, just substituting in an alternative probability distribution for Eve along with removing the constraint of $J_{i_t}$ in the post measured state $\ket{\mu(\ell, r, t)}$.  However, Equation \ref{eq:td}, along with the fact that quantum operations cannot increase trace distance, ensures that the real state $\rho^t_{ABE}$ and the above ideal state $\sigma_{ABE}^t$, are $\epsilon$-close in trace distance (averaged over all subset choices $t$).

\textbf{Step 2 - Defining ``Ideal-Ideal'' States: } Analyzing the entropy in the state $\sigma^t$, defined above, despite being composed of ideal states according to Theorem \ref{thm:sample}, is still challenging.  Instead, we will define a new ``ideal-ideal'' state, $\tau^t$, which is $\epsilon$ close in trace distance to $\sigma^t$ (and therefore $2\epsilon$ close to the original real state $\rho$).

For any fixed $i \in \B^m$, let's consider the expected value of $w(\ell^\phase\oplus r^\phase\oplus i^\phase)$.  Note that the $\ell$ and $r$ strings are chosen independently of $i$ since they are created from the honest but noisy sub-networks (or, rather, Heidi in our case, who is simulating the honest sub-networks).  Let $Y$ be the random variable taking the value $y=\ell^\phase \oplus r^\phase \in \{0,1\}^m$ with probability:
\[
  P_Y(y) = \sum_{\substack{\ell, r\in\B^m\\ \ell^\phase\oplus r^\phase = y}} P_L(\ell)P_R(r).
\]
We compute the expected value of $w(Y\oplus i^\phase)$ for a given $i\in\B^m$.  Recall from Definition \ref{def:noise-parameter}, the following:
\[
  p^* = \sum_{\substack{x,y\in\B\\x^\phase\oplus y^\phase = 1}}P_L(x)P_R(y),
\]
which, we assume, is known (or lower-bounded) by Alice and Bob.

Let $x = w(i^\phase)$.  We divide $i^\phase$ into two substrings: Even (all zeros) and Odd (all ones).  Since $Y$ is independent of $i$, one will expect $m\cdot x\cdot p^*$ ones to appear in the odd substring of $i^\phase$ while $m(1-x)p^*$ ones to appear in the even substring.  Note any one appearing in the odd substring of $i^\phase$ will cause $Y\oplus i^\phase$ to be zero (i.e., will decrease the number of ones in $Y\oplus i^\phase$) while any one in the even substring will cause $Y\oplus i^\phase$ to be one (i.e., will increase the number of ones in $Y\oplus i^\phase$).  Thus, the expected number of ones in $Y\oplus i^\phase$ for fixed $i\in\B^m$ is easily found to be:
\begin{align}
  \mu_i = \mathbb{E}(w(Y\oplus i^\phase) \text { } | \text{ } i) &= \underbrace{(x - xp^*)}_{\text{from odd substring}} + \underbrace{(1-x)p^*}_{\text{from even substring}} = x(1-p^*) + (1-x)p^*\notag\\\notag\\
  &= w(i^\phase)(1-p^*) + (1-w(i^\phase))p^*.
\end{align}
Now, by Hoeffding's inequality, and our choice of $\delta'$ from Equation \ref{eq:thm:deltap}, it holds that:
\begin{equation}\label{eq:pr-bound}
  Pr\left(|w(Y\oplus i^\phase) - \mu_i| \le \delta'\right) \ge 1-2\exp(-2(\delta')^2m) = 1-\epsilon,
\end{equation}
where the probability is over the outcome of $Y = \ell^\phase \oplus r^\phase$.

For any $i\in\B^m$, define the set $G_i$ to be:
\begin{equation}
  G_i = \{y \in \{0,1\}^m \st |w(y\oplus i^\phase) - \mu_i| \le \delta'\}.
\end{equation}
From this, we may write $\sigma^t$ as follows:
\begin{align}
  \sigma^t &= \frac{1}{4^{2m}}\sum_{z,u}\kb{\mbox{``}z,u\mbox{''}}_{cl}\otimes\left(\sum_{i\in\B^m}P_E(i)\otimes\sum_{y\in\{0,1\}^m}P_Y(y)\kb{y\oplus i^\phase}_P\otimes\sum_{\substack{\ell,r\in\B^m\\\ell^\phase\oplus r^\phase = y}}P_{LR}(\ell,r|y)\kb{\mu(\ell, i, r)}\right)\notag\\\notag\\
         &= \frac{1}{4^{2m}}\sum_{z,u}\kb{\mbox{``}z,u\mbox{''}}\sum_{i\in\B^m}P_E(i)\otimes\left(\sum_{y\in G_i}P_Y(y)\kb{y\oplus i^\phase}_P\otimes\sum_{\substack{\ell,r\in\B^m\\\ell^\phase\oplus r^\phase = y}}P_{LR}(\ell,r|y)\kb{\mu(\ell, i, r)}\right.\notag\\
  &\left.+ \sum_{y\not\in G_i}P_Y(y)\kb{y\oplus i^\phase}_P\otimes\sum_{\substack{\ell,r\in\B^m\\\ell^\phase\oplus r^\phase = y}}P_{LR}(\ell,r|y)\kb{\mu(\ell, i, r)}\right)
\end{align}
where we define:
\begin{equation}
  P_{LR}(\ell,r|y) = \frac{P_L(\ell)P_R(r)}{P_Y(y)}.
\end{equation}

Now, consider the state $\tau^t$ which consists only of ``good'' strings $y$:
\begin{equation}\label{eq:tau}
  \tau^t = \frac{1}{4^{2m}}\sum_{z,u}\kb{\mbox{``}z,u\mbox{''}}_{cl}\sum_{i\in\B^m}P_E(i)\otimes\left(\frac{1}{N_i}\sum_{y\in G_i}P_Y(y)\kb{y\oplus i^\phase}_P\otimes\sum_{\substack{\ell,r\in\B^m\\\ell^\phase+r^\phase = y}}P_{LR}(\ell,r|y)\kb{\mu(\ell, i, r)}\right),
\end{equation}
where, from Equation \ref{eq:pr-bound}, we have:
\begin{equation}
  N_i = \sum_{y\in G_i}P_Y(y)\sum_{\substack{\ell,r\in\B^m\\\ell^\phase\oplus r^\phase = y}}P_{LR}(\ell,r|y) \ge 1-\epsilon.
\end{equation}
(For the above, note that $\sum_{\substack{\ell,r\in\B^m\\\ell^\phase\oplus r^\phase = y}}P_{LR}(\ell,r|y) = 1$.)

Now, it is not difficult to show, using basic properties of trace distance and the triangle inequality, that the trace distance between $\tau^t$ and $\sigma^t$ is upper bounded by $\epsilon$.  Indeed:
\begin{align*}
  \frac{1}{2}\trd{\sigma^t-\tau^t} &\le \frac{1}{2}\sum_{z,u}\frac{1}{4^{2m}}\sum_i P_E(i)\left|\left| \left(1-\frac{1}{N_i}\right)\sum_{y \in G_i}P_Y(y)\kb{y+i^\phase}_P\otimes\sum_{\substack{\ell,r\in\B^m\\\ell^\phase+r^\phase = y}}P_{LR}(\ell,r|y)\kb{\mu(\ell, i, r)}\right.\right.\\
                               &\left.\left.+ \sum_{y \not\in G_i}P_Y(y)\kb{y+i^\phase}_P\otimes\sum_{\substack{\ell,r\in\B^m\\\ell^\phase+r^\phase \equiv y}}P_{LR}(\ell,r|y)\kb{\mu(\ell, i, r)}\right|\right|\\\\
  &\le \frac{1}{2}\sum_{z,u}\frac{1}{4^{2m}}\left(\sum_i P_E(i)\left(\frac{1}{N_i}-1\right)N_i + (1-N_i)\right) \le \epsilon
\end{align*}
Note that the above holds for any subset $t$.  Since we have that $\frac{1}{2}\trd{\sigma^t-\tau^t} \le \epsilon$ for every $t$, it follows from Equation \ref{eq:td} and the fact that CPTP maps cannot increase trace distance, that $\frac{1}{2}\trd{\sum_tP_T(t)\kb{t}(\rho^t-\tau^t)} \le 2\epsilon$.  Thus, we can actually analyze the entropy in $\tau^t$ and then use Lemma \ref{lemma:cptp-entropy} to promote the analysis to the real state.

\textbf{Step 3 - Bounding the Min Entropy: }
Rewriting $\tau^t$, Equation \ref{eq:tau}, we find:
\begin{align}
  \tau^t &= \sum_{Q_X\in\{0,1\}^m}\kb{Q_X}_P \otimes \sum_{i,y \in G^{(Q_X)}}\frac{P_E(i)P_Y(y)}{N_i} \otimes \sum_{z,u} \frac{1}{4^{2m}}\kb{\mbox{``}z,u\mbox{''}}\otimes\sum_{\ell,r}P_{LR}(\ell,r|y)\kb{\mu(\ell, i, r)},
\end{align}
where:
\begin{align}
  G^{(Q_X)} &= \left\{i\in\B^m,y\in\{0,1\}^m \st y\oplus i^\phase = Q_X \text{ and } y \in G_i\right\}\notag\\
          &= \left\{(i,Q_X\oplus i^\phase) \st Q_X\oplus i^\phase\in G_i\right\}\notag\\
          & \cong \left\{i \in \B^m \st \left|w(Q_X) - \mu_i\right|\le \delta'\right\}\notag\\
  &=\left\{i\in\B^m \st \left|w(Q_X) - (w(i^{\phase})(1-p^*) + (1-w(i^{\phase}))p^*)\right| \le \delta'\right\}\notag\\
  &= \left\{i \in \B^m \st \left|w(i^\phase) - \frac{w(Q_X) - p^*}{1-2p^*}\right| \le \frac{\delta'}{1-2p^*}\right\} = \widetilde{G}^{(Q_X)}.\label{eq:good-states-q}
\end{align}
With this, we can continue to manipulate the expression for $\tau^t$ as follows:
\begin{align}
  \tau^t &= \sum_{Q_X\in\{0,1\}^m}\kb{Q_X}_P \otimes \sum_{i \in \widetilde{G}^{(Q_X)}}\frac{P_E(i)P_Y(Q_X\oplus i^\phase)}{N_i} \otimes \sum_{z,u} \frac{1}{4^{2m}}\kb{\mbox{``}z,u\mbox{''}}\otimes\sum_{\ell,r}P_{LR}(\ell,r|Q_X\oplus i^\phase)\kb{\mu(\ell, i, r)},
\end{align}

At this point, Alice and Bob measure the $P$ register to observe an actual parity string $Q_X$ representing the observed $X$ basis noise.  The state then collapses to $\tau^{(t,Q_X)}$ which may be written as (tracing out the $P$ register which is simply $\kb{Q_X}_P$ after the measurement):
\begin{align}
  \tau^{(t,Q_X)} &= \frac{1}{M_{Q_X}}\sum_{i \in \widetilde{G}^{(Q_X)}} \frac{P_E(i)P_Y(Q_X\oplus i^\phase)}{ N_i} \otimes \sum_{z,u}\frac{1}{4^{2m}}\kb{\mbox{``}z,u\mbox{''}} \overbrace{\sum_{\ell,r}P(\ell,r|Q_X\oplus i^\phase)\kb{\mu(\ell, i, r)}}^{\tau^{(t,Q_X,z,u,i)}}\notag\\
  &= \sum_{i\in\widetilde{G}^{(Q_X)}}\widetilde{P}_E(i)\sum_{z,u}\frac{1}{4^{2m}}\kb{\mbox{``}z,u\mbox{''}}\otimes \tau^{(t,Q_X,z,u,i)}.
\end{align}
where, above $M_{Q_X}$ is a normalization term, and we define $\widetilde{P}_E(i) = P_E(i)P_Y(Q_X\oplus i^\phase)/(N_iM_{Q_X})$.  Note that each $\tau^{(t,Q_X,z,u,i)}$ is normalized.

At this point, the network operation $\net$ completes on the remaining unmeasured subset $-t$ (those qubits in the $\tau^{(t,Q_X,z,u,i)}$ system) and Alice and Bob measure their remaining systems in the $Z$ basis.  Given a particular observed $Q_X$, we need to bound $\Hmin(A|E)_{\tau^{(t,Q_X)}}$, where the $A$ register is taken to mean Alice's $Z$ basis outcome after the final network operation is performed.


Recalling the definition of $\ket{\mu(\ell, i, r)}$ (see Equation \ref{eq:mu}) we can write the state of $\tau^{(t,Q_X,z,u,i)}$, after the final network operation is performed on the remaining systems, as:
\begin{align}
&  \tau^{(t,Q_X,z,u,i)} = \sum_{\ell,r \in \B^m}P_{LR}(\ell,r|Q_X\oplus i^\phase)\notag\\
&\times P\left(\sum_{\ell_{-t},r_{-t}\in\B^n}\sqrt{P_L(\ell_{-t})P_R(r_{-t})}\ket{F_{\ell_{-t}},G_{r_{-t}}}\frac{1}{2^{2n}}\sum_{z_{-t},u_{-t}\in\B^n}\ket{\mbox{``}z_{-t},u_{-t}\mbox{''}}_{cl}\right.\notag\\
  &\otimes\left.\sum_{i_{-t} \in J_i}(-1)^{g(\ell_{-t},i_{-t},r_{-t};z_{-t},u_{-t})}\ket{\phi_{\ell_{-t} + i_{-t} + r_{-t}}}\ket{E_{i_{-t}|i_t}}\right)
\end{align}
Then, tracing out the remaining $G$ and $F$ registers while measuring the $cl$ register, yields the state:
\begin{align}
  \tau^{(t,Q_X,z,u,i)} &= \sum_{z_{-t},u_{-t}} \sum_{\ell_{-t},r_{-t} \in \B^n}\widetilde{P}_{CLR}(z_{-t},u_{-t},\ell_{-t},r_{-t})\kb{\mbox{``}z_{-t},u_{-t}\mbox{''}}\notag\\
  &\otimes\underbrace{P\left(\sum_{i_{-t} \in J_i}(-1)^{g(\ell_{-t}, i_{-t}, r_{-t}; z_{-t},u_{-t})}\ket{\phi_{\ell_{-t} + \i_{-t} + r_{-t}}}\ket{E_{i_{-t}|i_t}}\right)}_{\tau^{(t,Q_X,z,u,i,\ell,r)}}.\label{eq:tau-final}
\end{align}

Recall that, for any observed $Q_X$, it holds that $i \in \widetilde{G}^{(Q_X)}$ since $\tau$ is our ideal-ideal state.  Combining this knowledge, and using Equation \ref{eq:min-entropy-mixed}, for any $t$ and $Q_X$ we have:
\begin{equation}
  \Hmin(A|E)_{\tau^{(t,Q_X)}} \ge \min_{i\in\widetilde{G}^{Q_X}}\min_{z,u,\ell,r}\Hmin(A|E)_{\tau^{(t,Q_X,z,u,i,\ell,r)}}.
\end{equation}
Using Lemma \ref{lemma:bell-entropy}, along with Equation \ref{eq:tau-final}, we have:
\begin{equation}
  \min_{i\in\widetilde{G}^{Q_X}}\min_{z,u,\ell,r}\Hmin(A|E)_{\tau^{(t,Q_X,z,u,i,\ell,r)}} \ge \min_{i\in\widetilde{G}^{(Q_X)}}\left(n - n\bar{h}\left(w(i^\phase) + \delta\right)\right) = n\left(1 - \max_{i\in\widetilde{G}^{(Q_X)}}\bar{h}(w(i^\phase)+\delta)\right)
\end{equation}
Considering the definition of $\widetilde{G}^{(Q_X)}$ in Equation \ref{eq:good-states-q}, it is clear that, for any $Q_X$ and for any $i \in \widetilde{G}^{(Q_X)}$, it holds that:
\[
  w(i^\phase) \le \frac{w(Q_X) - p + \delta'}{1-2p}
\]
leading us to conclude that:
\begin{align}
  \Hmin(A|E)_{\tau^{(t,Q_X)}} &\ge n\left(1 - \bar{h}\left(\frac{w(q) - p + \delta'}{1-2p} + \delta\right)\right)
\end{align}

Now, of course this was only the ideal-ideal state, however it holds for any choice of $t$ and observed $Q_X$.  Since $\frac{1}{2}\trd{\sum_tP_T(t)\kb{t}(\rho^t-\tau^t)} \le 2\epsilon$, we can use Lemma \ref{lemma:cptp-entropy} to finish the proof.  Indeed, note that all operations performed on the two systems were CPTP maps satisfying the lemma's hypothesis, and we take the subset choice and observation of $Q_X$ as the random variable $X$ in that lemma.
\end{proof}

Our above theorem, along with Equation \ref{eq:PA}, will then allow us to determine a bound on the actual key-rate of the protocol. Setting $\epsilon_{PA} = 17\epsilon + 4(2\epsilon)^{1/3}$, for user specified $\epsilon$, then, the overall key-rate of the protocol will be:
\begin{equation}\label{eq:finite-rate}
  \texttt{rate} \ge \frac{N-m}{N}\left(1 - \bar{h}\left(\frac{w(Q_X) - p^* + \delta'}{1-2p^*} + \delta\right)\right) - \leakEC - \frac{1}{N}\log\frac{1}{\epsilon},
\end{equation}
except with a failure probability of at most $\epsilon_{fail} = 2(2\epsilon)^{1/3}$.  Above, $\leakEC$ is the error correction leakage.

The above is a finite-key bound on the key-rate of this QKD protocol operating over a partially corrupted repeater chain.  We may also easily use this to derive an asymptotic bound as shown in the following corollary:

\begin{corollary}
  In the asymptotic regime, where the number of signals $N$ approaches infinity, assuming the observed $Z$ and $X$ basis noise are identical (namely $Q_X$) the key-rate becomes:
  \begin{equation}\label{eq:asymptotic-rate}
    \texttt{rate}_{\infty} = 1 - \bar{h}\left(\frac{w(Q_X) - p^*}{1-2p^*}\right) - h(w(Q_X)).
  \end{equation}
\end{corollary}
\begin{proof}
  First, we may assume collective attacks, whereby the total signal is of the form $\rho_{ABE}^{\otimes N}$ and then, later, use de Finetti style arguments \cite{konig2005finetti} to promote the analysis to general attacks.  In the asymptotic setting, we may also set $m = \sqrt{N}$.  The result then follows from the quantum asymptotic equipartition property \cite{tomamichel2009fully} and the fact that $\frac{1}{N-m}\texttt{leak}_{EC}$ will approach $h(w(Q))$.
\end{proof}


\section{Evaluation} \label{sec:eval}

We now evaluate our key-rate bounds in both finite key (Equation \ref{eq:finite-rate}) and asymptotic (Equation \ref{eq:asymptotic-rate}) settings.  As we will soon see, our security model allows for significantly higher key-rates when some of the network can be assumed honest, when compared to the standard BB84 assumption, which assumes the entire network is adversarial.  This shows the benefit of incorporate knowledge of honest repeaters into QKD key-rate calculations.

Our evaluation setup will assume a repeater chain with five repeaters (six fiber links total) connecting Alice to Bob.  We will consider scenarios where the number of honest nodes varies from zero (a fully corrupted network) to four.  We will also assume depolarizing channels connect all honest  parties, and that the total observed noise can be modeled as a depolarizing channel.  The latter is an assumption made just for our evaluations in order to simulate reasonable values for the expected observed noise $w(Q_X)$.  That is, we are assuming the adversary's attack introduces noise that can be modeled as a depolarizing channel, just for the sake of evaluation.  Finally, given this setup, we will evaluate both the finite key and the asymptotic behavior of the partially corrupted repeater chain and compare with the standard BB84 expression (which assumes a fully corrupted network).  For BB84, we will compare with finite key (which we will denote \texttt{BB84-F}) and asymptotic (denoted \texttt{BB84-A}) where appropriate.

To evaluate key-rates, we need to compute what the expected value of $w(Q_X)$ would be in a general repeater chain scenario.  Then, we must determine reasonable lower-bounds on $p^*$ in the given adversarial model.  To do this, we designed a Python program, which simulates the action of a repeater chain, computing the overall expected noise, given the link level noise.  We simulate the case where each link suffers depolarizing noise which is a valid noise model for our network assumption. Such a depolarizing channel maps a two qubit state $\rho$ to:
\begin{equation}
  \mathcal{E}_q(\rho) = (1-q)\rho + \frac{q}{4}I,
\end{equation}
where $I$ is the identity operator on two qubits.
In particular, given the total number of repeaters $c$, and a link level noise parameter $q$, our program will compute the total expected noise in the network (i.e., the value of $w(Q_X)$).   We then use this simulator to determine a bound on $p^*$ by having it compute what the noise would be in a smaller chain.
  
Although each link can have a different depolarizing parameter, we assume for our evaluations that they are all identically $q$.   In our evaluations, we assume each round is independent and identical to the others; thus after a single round, the final state of the full network (after all Bell measurements and Pauli corrections are performed) is:
\begin{equation}
    \rho_{full} = \sum_{i \in B} P_{q,c}(i^\bit, i^\phase) \kb{\phi_{i}}.
  \end{equation}
The X-basis error rate of the full network is then the probability of a phase-flip occurring:
\begin{equation}
    w\left(Q_X\right) = P_{q,c}(0, 1) + P_{q,c}(1, 1).
  \end{equation}
  We use our Python program to determine the value of $P_{q,c}(\bit,\phase)$, thus allowing us to determine $w(Q_X)$, the expected total observed error rate in the chain.  This also allows us to determine $p^*$ by simulating smaller honest sub-networks separately using our Python simulator (i.e., by determining $P_{q,c'}$ for $c' < c$).


  Figure \ref{fig:noise depolarization} shows how the total observed noise accumulates in repeater networks of varying levels of corruption.  In networks of fixed size, a larger honest sub-network corresponds to a higher percentage of the total measured noise being natural instead of adversarial.  This means that less information is leaked to the adversary, and higher key-rates can be obtained.

  \begin{figure}
\centering
\includegraphics[width=0.6\textwidth]
{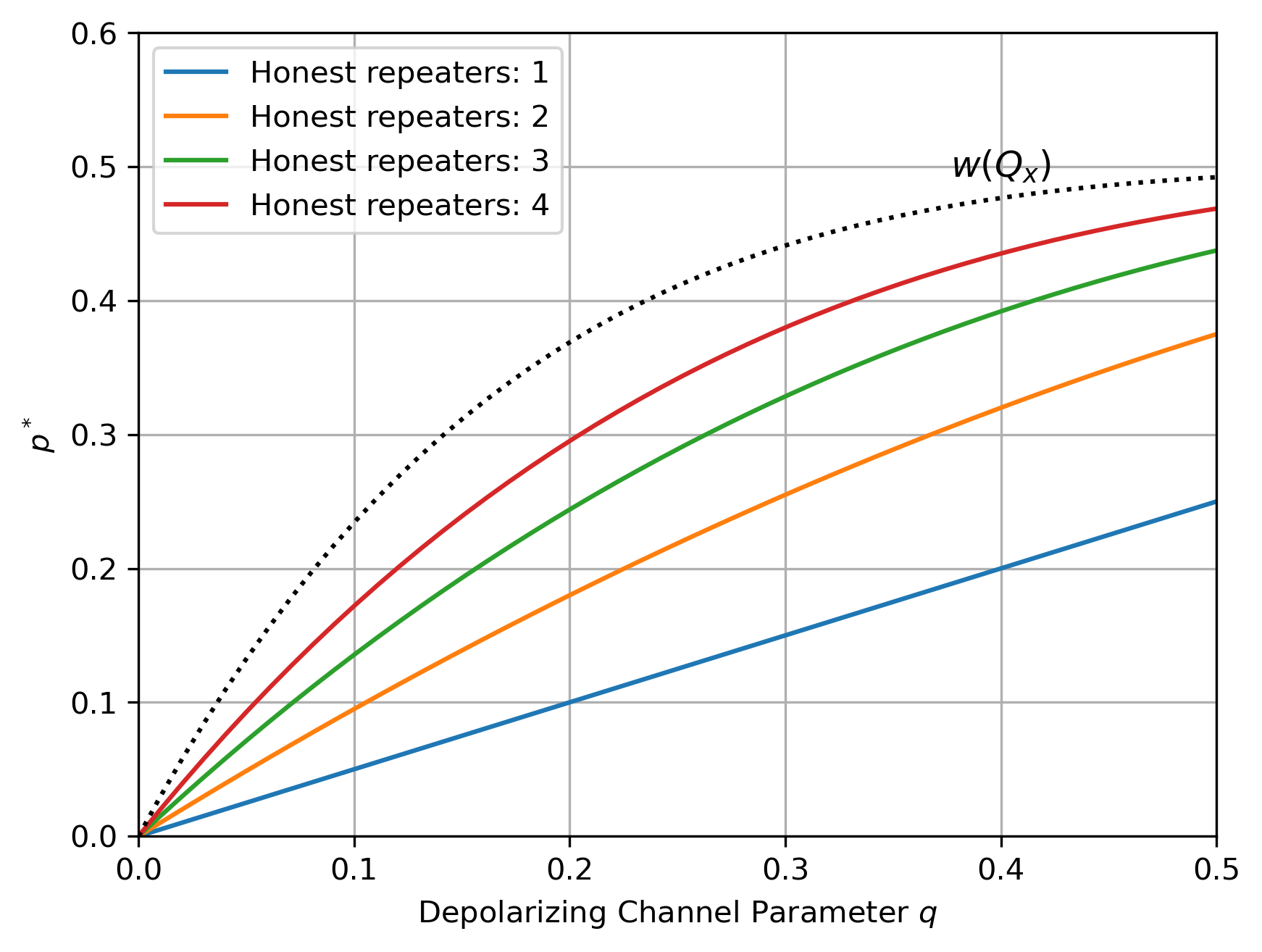}
\caption{Showing the total observed X-basis noise $w(Q_X)$ (top, Dashed line) and honest network noise parameters $p^*$ (Solid lines) as the link-level depolarizing noise $q$ increases.  Here we assume a five repeater (six link) network.  Noise parameter $p^*$ is computed for honest network sizes ranging from one to four repeaters.}
\label{fig:noise depolarization}
\end{figure}

We next evaluate our finite key-rate bound (Equation \ref{eq:finite-rate}), by setting $m=.07N$ and $\epsilon=10^{-36}$.  Such a setting for $\epsilon$ will imply a failure probability, and an $\epsilon_{PA}$-secure key (see Equation \ref{eq:PA}), both on the order of $10^{-12}$.  We assume $\leakEC$ is simply $1.2h(w(Q_X)+\delta)$.  We simulate networks with five total repeaters (six total links) and analyze the cases where zero, two, and four repeaters are honest.  Note that when the number of honest repeaters is zero, the entire network is assumed to be under adversarial control.

To benchmark our key rate against prior work, we use the finite BB84 key rate from \cite{tomamichel2012tight}, namely:
\begin{equation}
    \texttt{BB84-F} = \frac{N-m}{N}\left(1-\Bar{h}\left(w(Q_X) + \nu\right) - 1.2\Bar{h}\left(w(Q_X) + \nu\right)\right), 
\end{equation}
where
\begin{equation}
    \nu = \sqrt{\frac{N(m+1)\ln{(2/\epsilon)}}{m^2n}}.
\end{equation}

As seen in Figure \ref{fig:signals}, our finite key rate can significantly outperform standard BB84 (which assumes a completely adversarial network), so long as one is willing to assume at least some of the repeaters are honest but noisy. Note that, as seen in this figure, if one assumes the entire network is corrupted (thus $p^*=0$), our finite key rate requires more signals than BB84-F on fully corrupted networks to recover the same key-rate.  This is not unexpected, however, and is an artifact of our proof method.  In our proof, we need to sample twice to perform our ideal-ideal state analysis, thus giving worse bounds than BB84, with a fully adversarial network, where this is not required.  However, asymptotically the two key-rates (ours and standard BB84) converge.  When we take honest repeaters into account on partially corrupted networks, our key-rates outperform BB84-F, which is the entire point of our security proof.  Indeed, our results show that even assuming a small number of honest repeaters (even one next to Alice and Bob for instance), can lead to significant improvements in overall QKD performance.

\begin{figure}
\centering
\includegraphics[width=0.6\textwidth]{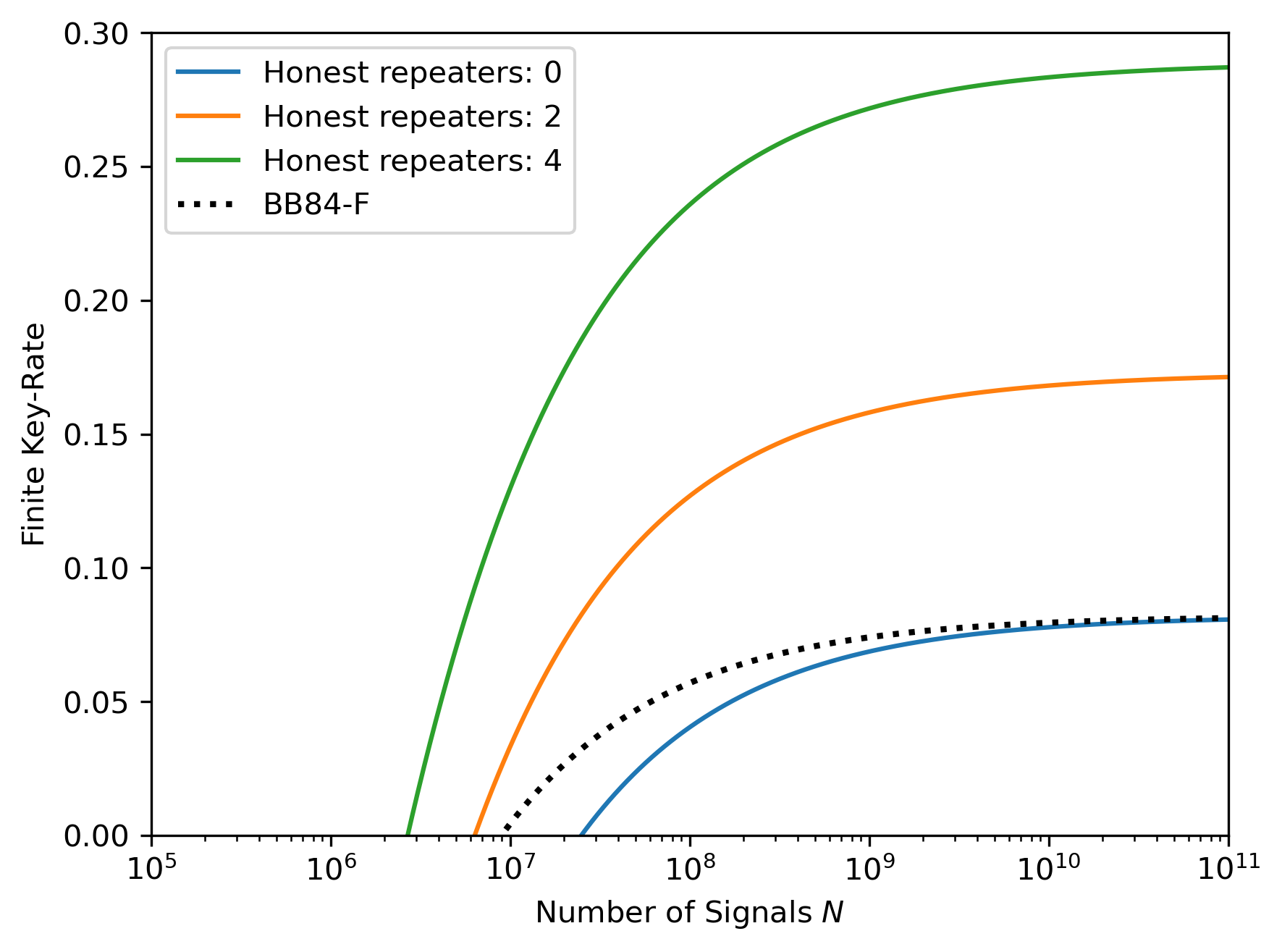}
\caption{Showing the finite key-rate ($y$-axis) as the number of signals $N$ increases ($x$-axis) in a five repeater chain for various numbers of honest repeaters (Solid lines). We compare to BB84 assuming a fully corrupted network (dashed line).  For this graph, we set the link-level noise in each link to be $q=3\%$.}
\label{fig:signals}
\end{figure}

Figure \ref{fig:finite noise} shows the noise tolerances of our finite key result, for $N=10^7$ and $10^8$ signal rounds.  On fully corrupted networks, we again see that BB84-F is more robust to total measured noise especially when less signals are transmitted (which is, again, an artifact of our proof method).  However, when assumptions of partial corruption can be made, our model provides higher key-rates and noise tolerances than BB84-F.

\begin{figure}
    \centering
    \begin{minipage}[b]{0.45\textwidth}
        \centering
        \includegraphics[width=\textwidth]{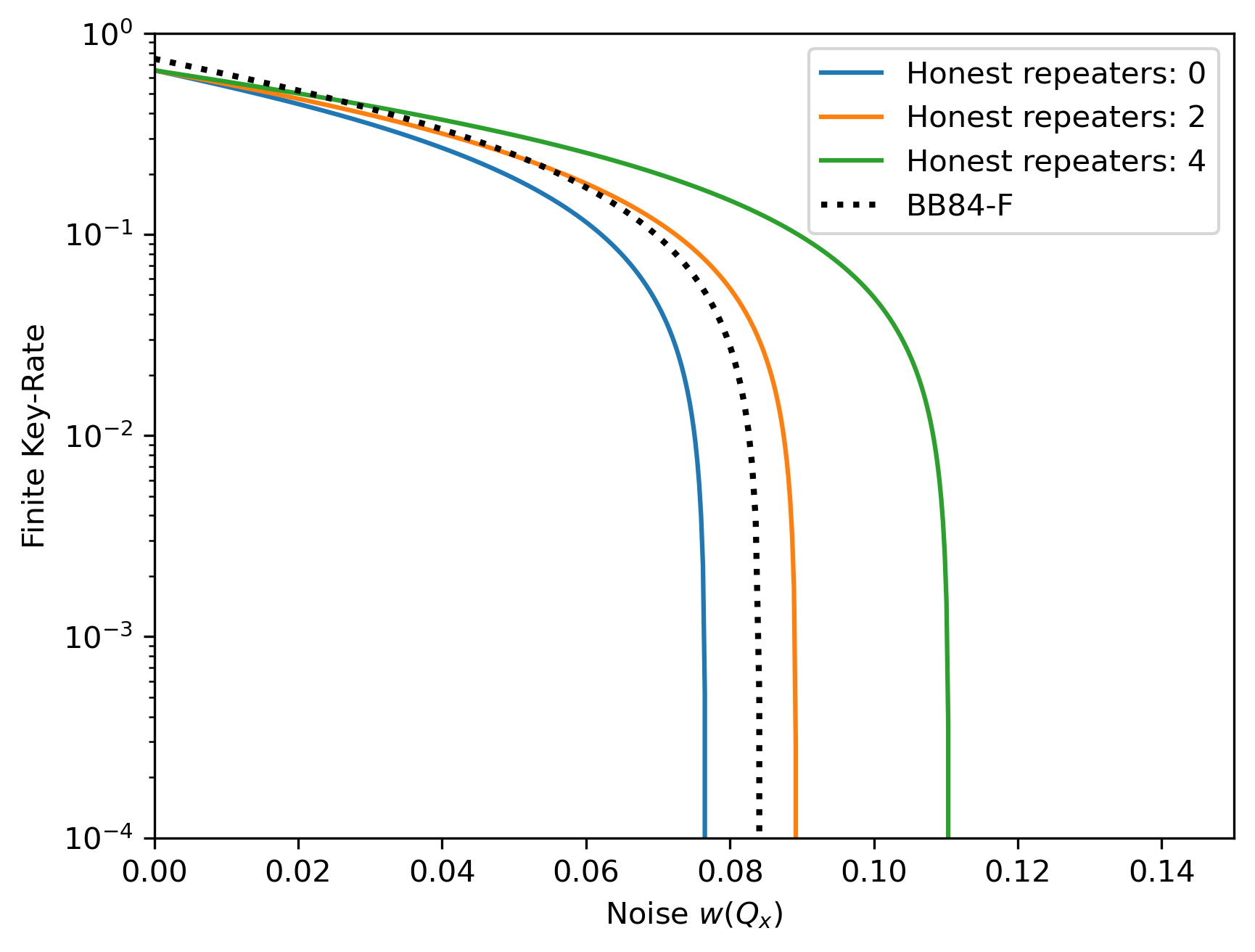}
    \end{minipage}
    \begin{minipage}[b]{0.45\textwidth}
        \centering
        \includegraphics[width=\textwidth]{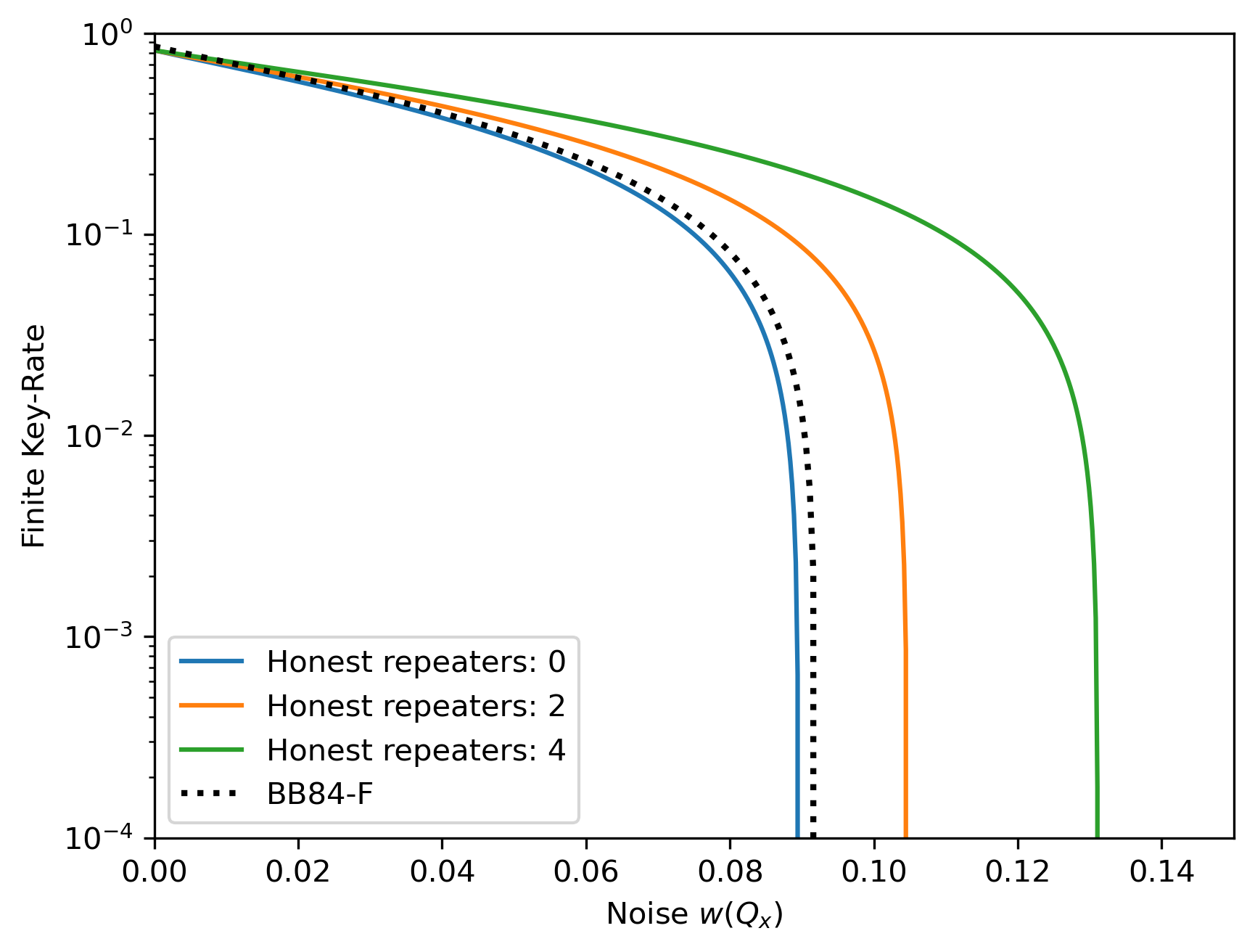}
    \end{minipage}
    \caption{Finite key-rates versus total observed X-basis noise in a five repeater network with $N=10^7$ signals (left) and $N=10^8$ signals (right).  Note that higher noise tolerances are possible when one assumes at least some of the repeaters are honest, but noisy.}
    \label{fig:finite noise}
\end{figure}

To compare our asymptotic key-rate (Equation \ref{eq:asymptotic-rate}) with BB84 in the standard security model, we use the well known BB84 rate from \cite{QKD-BB84-rate1}:
\begin{equation}
    \texttt{BB84-A} = 1 - 2\Bar{h}(w(Q_X)).
\end{equation}

We see in Figure \ref{fig:asymptotic noise} that our key-rates exactly align with BB84-A on fully corrupted networks.   On partially corrupted networks, our results produce higher key-rates for all noise levels with any non-zero number of honest repeaters.  This figure also shows that our result converges asymptotically to standard BB84 results, when one assumes a completely corrupted network (setting $p^*=0$).

\begin{figure}
\centering
\includegraphics[width=0.6\textwidth]
{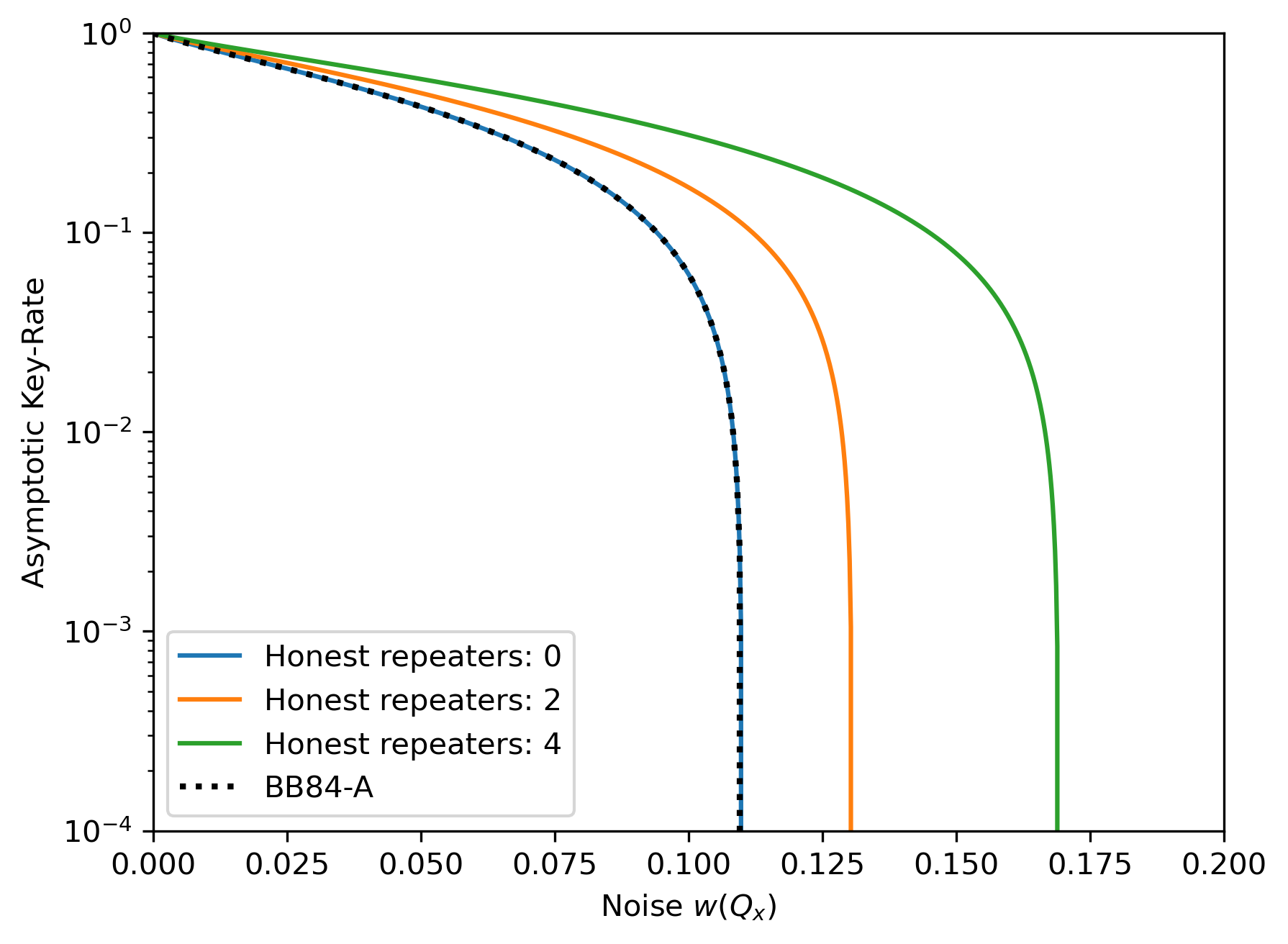}
\caption{Asymptotic key-rates versus total X-basis noise in a five repeater network.  With zero honest repeaters (i.e., a fully corrupted network), our key-rate result converges with asymptotic BB84.  As the number of honest repeaters in the network increases, our protocol outperforms asymptotic BB84, showing rigorously that increased noise tolerances and key-rates are possible in a partially corrupted setting.}
\label{fig:asymptotic noise}
\end{figure}

Taken together, these evaluations demonstrate that significantly increased noise tolerance and efficiency (i.e., increased key-rates) are possible when using our results and assuming at least some repeaters in a repeater chain are honest.

\section{Closing Remarks}

In this paper, we analyzed the case of a partially trusted (or partially corrupted) repeater chain.  Several motivating examples justify this security assumption: in particular, in a large-scale network, it is unreasonable to expect an adversary to be able to completely replace all devices on the network with perfect ones and, thus, ``hide'' within the expected natural noise.  There may also be cases where certain repeaters live in a safe area, where they can be trusted (as with trusted nodes).  While it is expected that higher key-rates are possible in this scenario, proving it in the finite-key scenario is non-trivial.  We derived a rigorous finite-key security proof for this setting; our proof techniques may be broadly applicable to other scenarios where there is a mix of trusted noise and adversarial noise in a quantum network.  We also evaluated our results, comparing with standard security assumptions and showed that, even with a small number of trusted repeaters, higher key-rates and noise tolerances are possible.  This shows the benefit in physically securing at least some portion of future quantum networks, perhaps those nodes and links near parties, even if one cannot secure the entire network from adversary attack.

Many interesting future problems remain.  First, it would be highly interesting to investigate lossy channels.  We suspect our proof method can be adapted to lossy channels, though we leave a rigorous proof as future work.  Investigating practical device imperfections would also be beneficial (e.g., multi-photon sources and imperfect detectors).  Finally, extending our proof to arbitrary networks, instead of just repeater chains, would be very interesting.


$ $\newline
\textbf{Acknowledgments: } AH was supported by the Defense Advanced Research Projects Agency (DARPA) ONISQ grant W911NF2010022, titled The Quantum Computing Revolution and Optimization: Challenges and Opportunities. WOK would like to acknowledge support from the NSF under grant number 2143644.


\end{document}